\documentclass[english]{article}
\pdfoutput=1

\usepackage[citecolor=blue]{hyperref}
\usepackage[latin9]{inputenc}
\usepackage{mathrsfs}
\usepackage{amsthm}
\usepackage{amsmath}
\usepackage{amssymb}
\usepackage{graphicx}
\usepackage{caption}
\usepackage{subcaption}
\usepackage{float}

\newcommand{\CA}{\mbox{$\mathsf{CA}$}}
\newcommand{\CAN}{\mbox{$\mathsf{CAN}$}}
\newcommand{\col}{\mbox{$\mathsf{col}$}}
\newcommand{\E}{\mbox{$\mathsf{E}$}}
\newcommand{\MCA}{\mbox{$\mathsf{MCA}$}}

\newcommand{\TS}{\mbox{$\mathsf{TS}$}}
\newcommand{\rand}{\mbox{$\mathsf{Rand}$}}
\newcommand{\mt}{\mbox{$\mathsf{MT}$}}
\newcommand{\naive}{\mbox{$\mathsf{Naive}$}}
\newcommand{\greedy}{\mbox{$\mathsf{Greedy}$}}
\newcommand{\color}{\mbox{$\mathsf{Col}$}}
\newcommand{\den}{\mbox{$\mathsf{Den}$}}
\newcommand{\T}{\mbox{$\mathsf{Trivial}$}}
\newcommand{\C}{\mbox{$\mathsf{Cyclic}$}}
\newcommand{\F}{\mbox{$\mathsf{Frobenius}$}}

\makeatletter

\theoremstyle{plain}
\newtheorem{thm}{\protect\theoremname}
\theoremstyle{plain}
\newtheorem{lem}[thm]{\protect\lemmaname}

\usepackage{fullpage}

\usepackage[lined, linesnumbered, ruled]{algorithm2e}

\makeatother

\usepackage{babel}
\providecommand{\lemmaname}{Lemma}
\providecommand{\theoremname}{Theorem}

\begin{document}

\title{Two-stage algorithms for covering array construction}

\author{Kaushik Sarkar and Charles J. Colbourn\\
School of Computing, Informatics, and Decision Systems Engineering\\
Arizona State University, PO Box 878809\\
Tempe, Arizona, 85287-8809, U.S.A.}

\maketitle

\begin{abstract}
Modern software systems often consist of many different components, each with a number of options. 
Although unit tests may reveal faulty options for individual components, functionally correct components may interact in unforeseen ways to cause a fault. 
Covering arrays are used to test for interactions among components systematically. 
A  two-stage framework, providing a number of concrete algorithms, is developed for the efficient construction of covering arrays. 
In the first stage, a time and memory efficient randomized algorithm covers most of the interactions. 
In the second stage, a more sophisticated search  covers the remainder in relatively  few tests. 
In this way, the storage limitations of the sophisticated search algorithms are avoided; 
hence the range of the number of components for which the algorithm can be applied is  extended, without increasing the number of tests. 
Many of the framework instantiations can be tuned to optimize a  memory-quality trade-off, so that fewer tests can be achieved using more memory.
The algorithms developed outperform the currently best known methods when the number of components ranges from 20 to 60, the number of options for each ranges from 3 to 6, and $t$-way interactions are covered for $t\in \{5,6\}$.
In some cases  a reduction in the number of tests  by more than $50\%$ is achieved.
\end{abstract}

\hspace{0.2in}\textbf{Keywords:} Covering array, Software interaction testing, Combinatorial construction algorithm

\section{Introduction}
Real world software and engineered systems are composed of many different components, each with a number of options,  that are required to work together in a variety of circumstances. 
Components are factors, and options for a component form the levels of its factor.
Although each level for an individual factor can be tested in isolation, faults in deployed software can arise from  interactions among levels of different factors. 
When an interaction involves levels of $t$ different factors, it is a $t$-way interaction.
Testing for faults caused by $t$-way interactions for every $t$ is generally infeasible, as a result of a combinatorial explosion.
However, empirical research on real world software systems indicates that testing all possible 2-way (or 3-way) interactions would  detect $70\%$ (or $90\%$) of all faults \cite{KUHN1}. 
Moreover, testing all possible 6-way interactions is sufficient for detection of $100\%$ of all  faults in the systems examined in \cite{KUHN1}. 
Testing all possible $t$-way interactions  for some $2 \leq t \leq 6$ is  \emph{pseudo-exhaustive testing} \cite{Kuhnbook}, and is accomplished with a combinatorial array known as a covering array.

Formally, let $N,\, t,\, k,$ and $v$ be integers with $k\ge t\ge 2$ and $v \ge 2$. 
A \emph{covering array} $\CA(N;t,k,v)$ is an $N\times k$ array $A$ in which each entry is from a $v$-ary alphabet $\Sigma$, and for every $N\times t$ sub-array $B$ of $A$ and every  ${\bf x} \in \Sigma^{t},$ there is a  row of $B$ that equals $\bf x$. 
Then $t$ is  the \emph{strength} of the covering array, $k$ is  the number of \emph{factors}, and $v$ is  the number of \emph{levels}. 

When $k$ is a  positive integer,  $[k]$  denotes the set $\{1,\ldots,k\}$. 
A \emph{$t$-way interaction}  is $\{(c_{i},a_{i})\,:\,1\le i\le t,\, c_{i}\in[k],\, c_{i}\neq c_{j}\,\text{for }i\neq j,\,\text{and }a_{i}\in\Sigma\}$.
So an interaction is an assignment of levels from  $\Sigma$ to  $t$  of the $k$  factors. 
$\mathcal{I}_{t,k,v}$ denotes the set of all $\binom{k}{t}v^t$ interactions for given $t,\,k$ and $v$.
An $N\times k$ array $A$ \emph{covers} the interaction $\iota = \{(c_{i},a_{i})\,:\,1\le i\le t,\, c_{i}\in[k],\, c_{i}\neq c_{j}\,\text{for }i\neq j,\,\text{and }a_{i}\in\Sigma\}$ if there is a row $r$ in $A$ such that $A(r,c_{i})=a_{i}$ for $1\le i\le t$. 
When there is no such row in
$A$,  $\iota$ is \emph{not covered} in $A$. 
Hence a $\CA(N;t,k,v)$ covers all  interactions in $\mathcal{I}_{t,k,v}$.

Covering arrays are used extensively for interaction testing in complex engineered systems. To ensure that all  possible combinations of options of $t$ components function together correctly, one needs examine all possible $t$-way interactions. 
When the number of components is $k$, and the number of different options available for each component is  $v$, each row of $\CA(N;t,k,v)$ represents a test case. 
The $N$ test cases  collectively test all  $t$-way interactions. 
For this reason, covering arrays have been used in combinatorial interaction testing in varied fields like software and hardware engineering, design of composite materials, and biological networks \cite{CAWSE,Kuhnbook,KUHN2,RonC,SerB}. 

The cost of testing is directly related to the number of test cases. 
Therefore, one is interested in covering arrays with the fewest rows.
The smallest value of $N$ for which $\CA(N;t,k,v)$ exists is denoted by $\CAN(t,k,v)$. 
Efforts to determine or bound $\CAN(t,k,v)$ have been extensive; see \cite{sicily,croatia,Kuhnbook,NieL-CS} for example. 
Naturally one would prefer to determine $\CAN(t,k,v)$ exactly.  
Katona \cite{katona} and Kleitman and Spencer \cite{kleitman} independently showed that for $t=v=2$,
the minimum number of rows $N$ in a $\CA(N;2,k,2)$ is the smallest $N$ for which 
$k\leq{N-1 \choose \lceil\frac{N}{2}\rceil}$.
Exact determination of $\CAN(t,k,v)$ for other values of $t$ and $v$ has remained open. 
However, some progress has been made in determining upper bounds for $\CAN(t,k,v)$ in the general case; for recent results, see  \cite{sarkar16}.

For practical applications such bounds are often unhelpful, because one needs explicit covering arrays to use as test suites. 
Explicit constructions can be recursive, producing larger covering arrays using smaller ones as ingredients (see \cite{croatia} for a survey), or direct.
Direct methods for some specific cases arise from algebraic, geometric, or number-theoretic techniques; general direct methods are computational in nature.
Indeed when $k$ is relatively small, the best known results arise from computational techniques \cite{catables}, and these are in turn essential for the successes of recursive methods.
Unfortunately, the existing computational methods encounter difficulties as $k$ increases, but is still within the range needed for practical applications.
Typically such difficulties arise either as a result of storage or time limitations or by producing covering arrays that are too big to compete with those arising from simpler recursive methods.

Cohen \cite{Cohen04} discusses commercial software where the number of factors often exceeds $50$. 
Aldaco et al. \cite{Aldaco15} analyze a complex engineered system having 75 factors, using a variant of covering arrays. 
Android \cite{android16} uses a  \emph{Configuration} class to describe the device configuration; there are  $17$ different configuration parameters  with $3-20$ different levels. 
In each of these cases, while existing techniques are effective when the strength is small, these moderately large values of $k$ pose concerns for larger strengths.

In this paper, we focus on situations in which every factor has the same number of levels.  
These cases have been most extensively studied, and hence provide a basis for making comparisons.
In practice, however, often different components have different number of levels, which is captured by extending the notion of a covering array. 
A \emph{mixed covering array} $\MCA(N;t,k,(v_1,v_2,\ldots,v_k))$ is an $N\times k$ array in which the $i$th column contains $v_i$ symbols for $1 \leq i \leq k$.
When $\{i_1,\ldots,i_t\}\subseteq \{1,\ldots,k\}$ is a set of  $t$ columns, in the $N\times t$ subarray obtained by selecting columns $i_1,\ldots,i_t$ of the MCA, 
each of the  $\prod_{j=1}^t v_{i_j}$ distinct $t$-tuples appears as a row at least once.
Although we examine the uniform case in which $v_1 = \cdots =  v_k$, the methods developed here can all be directly applied to mixed covering arrays as well.

Inevitably, when $k > \max(t+1,v+2)$, a covering array must cover some interactions more than once, for if not they are orthogonal arrays \cite{oas}.
Treating the rows of a covering array in a fixed order, each row covers some number of interactions not covered by any earlier row.  
For a variety of known constructions, the initial rows cover many new interactions, while the later ones cover very few \cite{BryceC}.
Comparing this rate of coverage for a purely random method and for one of the sophisticated search techniques, one finds little difference in the initial rows, but very substantial differences in the final ones.  
This suggests strategies to build the covering array in stages, investing more effort as the number of remaining uncovered interactions declines.

In this paper we propose a new algorithmic framework for covering array construction,  the \emph{two-stage framework}. 
In the first stage, a randomized row construction method  builds a specified number of rows to cover all but at most a specified, small number of interactions. 
As we  see later, by dint of being randomized  this  uses very little memory. 
The second stage is  a more sophisticated search  that adds few rows  to cover the remaining uncovered interactions. 
We choose search algorithms whose requirements depend on the number of interactions to be covered, to profit from the fact that few interactions remain. 
By mixing randomized and deterministic methods, we hope to retain the fast execution and small storage of the randomized methods, along with the accuracy of the deterministic search techniques.
 
We introduce a number of algorithms within the two-stage framework. 
Some  improve upon best known bounds on $\CAN(t,k,v)$ (see \cite{sarkar16}) in principle.  But our focus is on the practical consequences: 
The two-stage algorithms are indeed quite efficient for  higher strength ($t\in\{5,6\}$) and moderate number of levels ($v\in\{3,4,5,6\}$), when the number of factors $k$ is moderately high (approximately in the range of $20-80$ depending on value of $t$ and $v$). 
In fact, for many combination of $t, k$ and $v$ values the two-stage algorithms beat the previously best known bounds. 

Torres-Jimenez et al. \cite{jimenez-2ssa} explore a related two-stage strategy. 
In their first stage, an in-parameter-order greedy strategy (as used in ACTS \cite{Kuhnbook}) adds a column to an existing array; in their second stage, simulated annealing is applied to cover the remaining interactions.  
They apply their methods when $t=v=3$, when the storage and time requirements for both stages remain acceptable.
In addition to the issues in handling larger strengths, their methods provide no \emph{a priori} bound on the size of the resulting array.  
In contrast with their methods, ours provide a guarantee prior to execution with much more modest storage and time.

The rest of the paper is organized as follows.
Section \ref{sec:review} reviews algorithmic methods of covering array construction, specifically the randomized algorithm and the density algorithm. 
This section contrasts these two methods and points out their limitations. 
Then it gives an intuitive answer to the question of why a two stage based strategy might work and introduces the general two-stage framework. 
Section \ref{sec:two-stage} introduces some specific two-stage algorithms. 
Section \ref{subsec:simple} analyzes and evaluates the na\"ive strategy. 
Section \ref{subsec:density} describes a two-stage algorithm that combines the randomized and the density algorithm. 
Section \ref{subsec:coloring} introduces graph coloring based techniques in the second stage. 
Section \ref{subsec:group} examines the effect of group action on the size of the constructed covering arrays. 
Section \ref{sec:results} compares the results of various two-stage algorithms with the presently best known sizes. 
In Section \ref{sec:mt-algo} we discuss the Lov\'asz local lemma (LLL)  bounds on $\CAN(t,k,v)$ and provide a Moser-Tardos type randomized algorithm for covering array construction that matches the bound. 
Although the bound was known \cite{GSS}, the proof was non-constructive, and a constructive algorithm to match this bound seems to be absent in the literature. 
We explore potentially better randomized algorithms for the first stage using LLL based techniques,
We also obtain a two-stage  bound that improves the LLL  bound for $\CAN(t,k,v)$.
We conclude the paper in Section \ref{sec:conc}.

\section{Algorithmic construction of covering arrays}\label{sec:review}

Available algorithms for the construction of covering arrays are primarily heuristic in nature; indeed exact algorithms have succeeded for very few cases.  
Computationally intensive metaheuristic search methods such as simulated annealing, tabu search, constraint programming, and genetic algorithms have been employed when the
strength is relatively small or the number of factors and levels is small.  
These methods have established many of the best known bounds on sizes of covering arrays \cite{catables}, but for many problems of practical size their time and storage requirements are prohibitive.  
For larger problems, the best available methods are greedy.
The IPO family of algorithms \cite{Kuhnbook} repeatedly adds one column at a time, and then adds new rows to ensure complete coverage. 
In this way, at any point in time, the status of $v^t \binom{k-1}{t-1}$ interactions may be stored.
AETG \cite{AETG} pioneered a different method, which greedily selects one row at a time 
to cover a large number of as-yet-uncovered interactions.  
They establish that if a row can be chosen that covers the maximum number, a good \emph{a priori} bound on the size of the covering array can be computed.  
Unfortunately selecting the maximum is NP-hard, and even if one selects the maximum there is no guarantee that the covering array is the smallest possible \cite{BryceC}, so AETG resorts to a good heuristic selection of the next row by examining the stored status of $v^t \binom{k}{t}$ interactions.
None of the methods so far mentioned therefore guarantee to reach an \emph{a priori} bound.
An extension of the AETG strategy, the density algorithm  \cite{DDA, nDDA, ColCECA}, stores additional statistics for all $v^t \binom{k}{t}$ interactions in order to ensure the selection of a good next row, and hence guarantees to produce an array with at most the precomputed number of rows. 
Variants of the density algorithm have proved to be most effective for problems of moderately large size.
For even larger problems, pure random approaches have been applied.

To produce methods that provide a guarantee on size, it is natural to focus on the density algorithm in order to understand its strengths and weaknesses. 
To do this,  we contrast it with a  basic randomized algorithm. 
Algorithm \ref{algo:rand-slj} shows a simple randomized algorithm for covering array construction. 
The algorithm constructs an array of a particular size randomly and checks whether all the interactions are covered.  It repeats until it finds an array that covers all the interactions.

\begin{algorithm}[t]
\SetKw{Break}{break}
\KwIn{$t$ : strength of the  covering array, $k$ : number of factors, $v$ : number of levels for each factor}
\KwOut{$A$ : a $\CA(N;t,k,v)$}
Set $N :=\left \lfloor  \frac{\log{k \choose t}+t\log v}{\log\left(\frac{v^{t}}{v^{t}-1}\right)} \right \rfloor$\;
\Repeat {covered $=$ true} { \label{algLine:iscovering}
    Construct an $N \times k$ array $A$ where each entry is chosen independently and uniformly at random from a $v$-ary alphabet\; \label{algLine:arr-construct}
    Set \emph{covered}$:=$ true\;
    \For {each interaction $\iota \in \mathcal{I}_{t,k,v}$}{ \label{algline:checkInt}
        \If {$\iota$ is not covered in $A$} {
            Set \emph{covered}$:=$ false\;
            \Break\;
        }
    }
}
Output $A$\;

\caption{A randomized algorithm for covering array construction.}
\label{algo:rand-slj}
\end{algorithm}

A $\CA(N;t,k,v)$ with $N = \frac{\log{k \choose t}+t \log v}{\log\left(\frac{v^{t}}{v^{t}-1}\right)}$ is guaranteed to exist:

\begin{thm}
\label{thm:slj} {\rm \cite{Johnson74,Lovasz75,Stein74}} (Stein-Lovász-Johnson (SLJ) bound): Let
$t,\, k,\, v$ be integers with $k\ge t\ge2$, and $v\ge2$. Then
as $k\rightarrow\infty$,

\[
\CAN(t,k,v)\le\frac{\log{k \choose t}+t \log v}{\log\left(\frac{v^{t}}{v^{t}-1}\right)}
\]
    
\end{thm}

In fact, the probability that the $N \times k$ array constructed in line \ref{algLine:arr-construct} of Algorithm \ref{algo:rand-slj} is a valid covering array is high enough that the expected number of times the loop in line \ref{algLine:iscovering} is repeated is a small constant.

%
%

An alternative strategy is to add rows one by one  instead of constructing the full array at the outset. 
We start with an empty array, and whenever we add a new row  we ensure that it covers at least the expected number of previously uncovered interactions for a randomly chosen row. 
The probability that an uncovered interaction is covered by a random row is $1/v^t$. 
If the number of uncovered interactions is $u$, then by linearity of expectation, the expected number of newly covered interactions in a randomly chosen row is $uv^{-t}$. 
If each  row  added  covers exactly this expected number, we obtain the same number of rows as the SLJ bound, realized in Algorithm \ref{algo:rand-slj}. 
But because the actual number of newly covered interactions is always an integer, each added row covers at least $\lceil uv^{-t} \rceil$ interactions. 
This is especially helpful towards the end when the expected number is a small fraction. 

Algorithm \ref{algo:discrete-slj} follows this strategy. 
Again the probability that a randomly chosen row covers at least the expected number of previously uncovered interactions is high enough that the expected number of times the row selection loop in line \ref{algLine:rand-row-select} of Algorithm \ref{algo:discrete-slj} is repeated is bounded by a small constant.

\begin{algorithm}[htb]
\SetKw{Break}{break}
\KwIn{$t$ : strength of the  covering array, $k$ : number of factors, $v$ : number of levels for each factor}
\KwOut{$A$ : a $\CA(N;t,k,v)$}
Let $A$ be an empty array\;
Initialize a table $T$ indexed by all ${k \choose t} v^t$ interactions, marking  every interaction \emph{``uncovered''}\;
\While {there is an interaction marked ``uncovered'' in  $T$} {
    Let $u$ be the number of interactions marked \emph{``uncovered''} in  $T$\;
    Set \emph{expectedCoverage} $:= \lceil \frac{u}{v^t} \rceil$\;
    \Repeat {coverage $>$ expectedCoverage} { \label{algLine:rand-row-select}
        Let $r$ be a row of length $k$ where each entry is chosen independently and uniformly at random from a $v$-ary alphabet\;
        Let \emph{coverage} be the number of \emph{``uncovered''} interactions in  $T$ that are covered in row $r$\;
    }
    Add $r$ to $A$\;
    Mark all interactions covered by $r$ as \emph{``covered''} in  $T$\;
}
Output $A$\;

\caption{A randomized algorithm for covering array construction using the discrete SLJ strategy.}
\label{algo:discrete-slj}
\end{algorithm}

We can obtain an upper bound on the size  produced by Algorithm \ref{algo:discrete-slj} by assuming that each new row  added  covers exactly $\lceil uv^{-t} \rceil$ previously uncovered interactions. 
This bound is the \emph{discrete Stein-Lovász-Johnson} (discrete SLJ) bound. 
Figure \ref{fig:comp-slj-d-slj} compares the sizes of covering arrays obtained from the SLJ and the discrete SLJ bounds for different values of $k$ when $t=6$ and $v=3$. 
Consider a concrete example, when $t=5,\,k=20$, and $v=3$. 
The SLJ bound guarantees the existence of a covering array with $12499$ rows, whereas the discrete SLJ bound guarantees the existence of a covering array with only $8117$ rows. 

\begin{figure}[htbp]
\begin{centering}
    \includegraphics[bb=100bp 238bp 520bp 553bp,clip,scale=0.8]{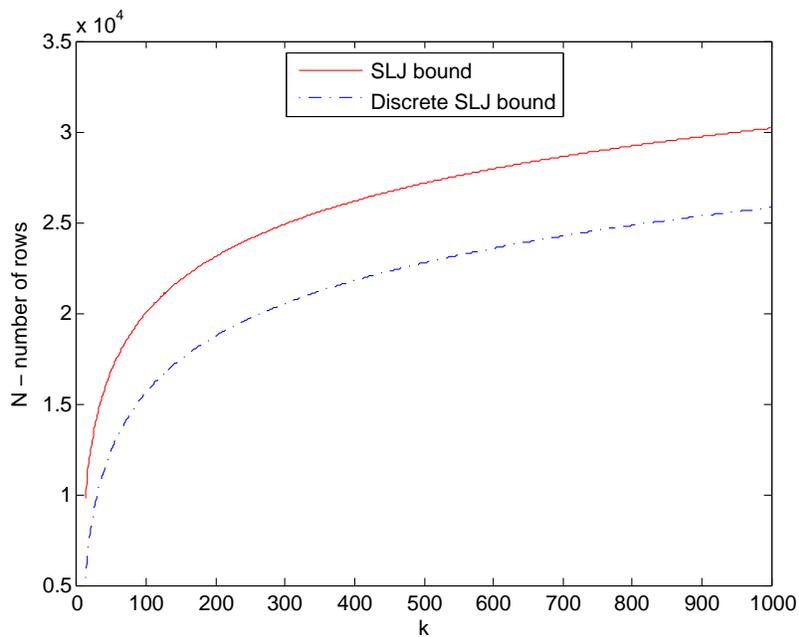}
    \par\end{centering}
    
    \caption{Comparison of covering array sizes obtained from SLJ bound and discrete SLJ bound for different values of $k$, when $t=6$ and $v=3$.}
    \label{fig:comp-slj-d-slj}
\end{figure}

The density algorithm replaces the  loop at line \ref{algLine:rand-row-select} of Algorithm \ref{algo:discrete-slj} by a conditional expectation  derandomized method. 
For fixed $v$ and $t$ the density algorithm selects a row efficiently (time polynomial in $k$) and deterministically that is guaranteed to cover at least $\lceil uv^{-t} \rceil$ previously uncovered interactions.
In practice, for small values of $k$ the density algorithm works quite well, often covering many more interactions than the minimum. 
Many of the currently best known $\CAN(t,k,v)$ upper bounds are obtained by the density algorithm in combination with various post-optimization techniques \cite{catables}. 

However, the practical applicability of Algorithm \ref{algo:discrete-slj}  and the density algorithm is limited by the  storage of the table $T$, representing  each of the $\binom{k}{t}v^t$ interactions. 
Even when $t=6$, $v=3$, and $k=54$, there are 18,828,003,285 6-way interactions. 
This huge memory requirement renders the density algorithm impractical for rather small values of $k$ when  $t\in\{5,6\}$ and $v\ge3$.
We present  an idea to circumvent this large requirement for memory, and develop it in full in Section \ref{sec:two-stage}.

\subsection{Why does a two stage based strategy make sense?}
Compare the two extremes,  the density algorithm and Algorithm \ref{algo:rand-slj}. 
On one hand, Algorithm \ref{algo:rand-slj} does not suffer from any substantial storage restriction, but appears to generate many more rows than the density algorithm. 
On the other hand, the density algorithm constructs fewer rows for small values of $k$, but becomes impractical  when $k$ is moderately large. 
One wants algorithms that behave like Algorithm \ref{algo:rand-slj} in terms of memory, but yield a number of rows competitive with the density algorithm.

For $t=6$, $k=16$, and $v=3$, Figure \ref{fig:cov-prof} compares the coverage profile for  the density algorithm and Algorithm \ref{algo:rand-slj}. 
We plot the number of newly covered interactions for each  row in the density algorithm, and  the expected number of newly covered interactions for each row  for Algorithm \ref{algo:rand-slj}. 
The qualitative features exhibited by this plot are representative of the rates of coverage for other parameters.

\begin{figure}[htbp]
\begin{centering}
\includegraphics[bb=100bp 238bp 520bp 553bp,clip,scale=0.8]{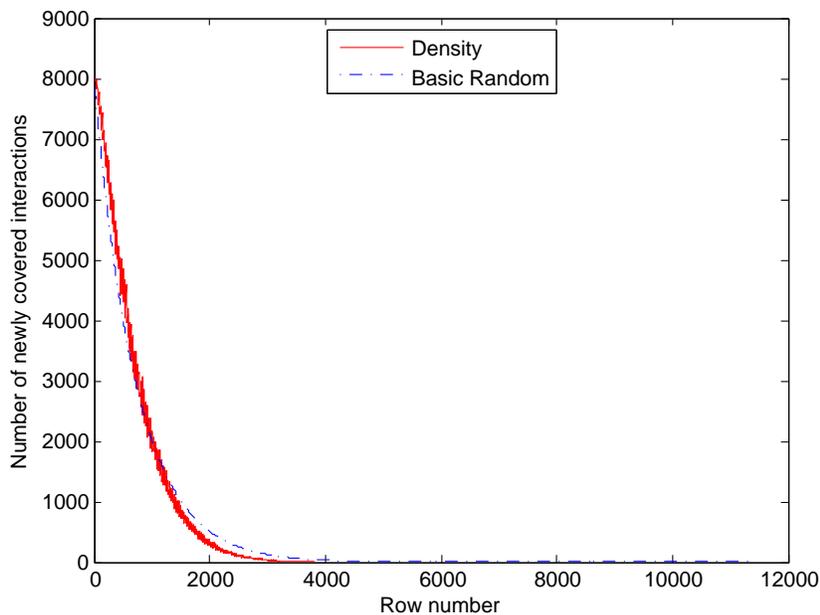}
\par\end{centering}

\caption{For $t=6$, $k=16$ and $v=3$, the actual number of newly covered interactions of  the density algorithm and the expected number of newly covered interactions in a random array.}
\label{fig:cov-prof}
\end{figure}

Two key observations are suggested by Figure \ref{fig:cov-prof}. 
First, the expected coverage in the initial random rows is similar to the rows chosen by the density algorithm. 
In this example, the partial arrays consisting of the first 1000 rows exhibit similar coverage, yet the randomized algorithm needed no extensive bookkeeping.
Secondly, as later rows are added, the judicious selections of the density algorithm produce much larger coverage per row than Algorithm \ref{algo:rand-slj}.
Consequently it appears sensible to invest few computational resources on the initial rows, while making more careful selections in the later ones.
This  forms the blueprint of our general \emph{two-stage} algorithmic framework shown in Algorithm \ref{algo:general-two-stage}.
 
\begin{algorithm}[htb]
\SetKw{Break}{break}
\KwIn{$t$ : strength of the required covering array, $k$ : number of factors, $v$ : number of levels for each factor}
\KwOut{$A$ : a $\CA(N;t,k,v)$}
Choose a number $n$ of rows and a number $\rho$ of  interactions\;
\tcp{First Stage}
Use a randomized algorithm to construct an $n \times k$ array $A'$\; \label{line:first-stage}
Ensure that $A'$ covers all but at most  $\rho$  interactions\;
Make a list $L$ of interactions that are not covered in $A'$ ($L$ contains at most $\rho$ interactions)\;
\tcp{Second Stage}
Use a deterministic  procedure to add $N-n$ rows to $A'$ to cover all the interactions in $L$\; \label{line:second-stage}
Output $A$\;

\caption{The general two-stage framework for covering array construction.}
\label{algo:general-two-stage}
\end{algorithm}

A specific covering array construction algorithm results by specifying the  randomized  method in the first stage, the  deterministic method in the second stage, and the computation of  $n$ and  $\rho$. 
Any such algorithm  produces a covering array, but we wish to make selections so that the resulting algorithms are practical while still providing a guarantee on the size of the array.
In Section \ref{sec:two-stage} we  describe different algorithms from the two-stage  family, determine the size of the partial array to be constructed in the first stage, and establish upper bound guarantees.
In Section \ref{sec:results} we explore how good the algorithms are in practice. 

\section{Two-stage framework} \label{sec:two-stage}

For the first stage we consider two methods:
\begin{center}
\begin{tabular}{rl}
\rand &the basic randomized algorithm\\
\mt &the Moser-Tardos type algorithm
\end{tabular}
\end{center}
We defer the development of method \mt\ until Section \ref{sec:mt-algo}. Method \rand\ uses a simple variant of Algorithm \ref{algo:rand-slj}, choosing a random $n \times k$ array.

For the second stage we consider four methods:
\begin{center}
\begin{tabular}{rl}
\naive&the na\"ive strategy, one row per uncovered interaction\\
\greedy&the online greedy coloring strategy\\
\den&the density algorithm\\
\color &the graph coloring algorithm
\end{tabular}
\end{center}
Using these abbreviations, 
we adopt a uniform naming convention for the algorithms: $\TS\left<A,B\right>$ is the algorithm in which $A$ is used in the first stage, and $B$ is  used in the second stage.
For example, $\TS\left<\mathrm{\mt,\greedy}\right>$ denotes a two-stage algorithm where the first stage is a Moser-Tardos type randomized algorithm and the second stage is a greedy coloring algorithm. 
Later when the need  arises we refine these algorithm names. 

\subsection{One row per uncovered interaction in the second stage ($\TS\left<\mathrm{\rand,\naive}\right>$)}\label{subsec:simple}
In the second stage  each of the uncovered interactions after the first stage  is covered using a new row. 
Algorithm \ref{algo:two-stage-simple} describes the method in more detail. 

\begin{algorithm}[htbp]
\SetKw{Break}{break}
\KwIn{$t$ : strength of the  covering array, $k$ : number of factors, $v$ : number of levels for each factor}
\KwOut{$A$ : a $\CA(N;t,k,v)$}
Let $n := \frac{\log{k \choose t}+t\log v+\log\log\left(\frac{v^{t}}{v^{t}-1}\right)}{\log\left(\frac{v^{t}}{v^{t}-1}\right)}$\; \label{algline:set-n}
Let $\rho = \frac{1}{\log\left(\frac{v^{t}}{v^{t}-1}\right)}$\; \label{algline:set-mu}

\Repeat {\emph{covered}$=$ true} { \label{algLine:is-enough-covering}
    Let $A$ be an $n \times k$ array where each entry is chosen independently and uniformly at random from a $v$-ary alphabet\;
    Let \emph{uncovNum} $:=0$ and  \emph{unCovList} be an empty list of interactions\;
    Set \emph{covered}$:=$ true\;
    \For {each interaction $\iota \in \mathcal{I}_{t,k,v}$}{ \label{algline:check-Int-loop}
        \If {$\iota$ is not covered in $A$} { \label{algline:check-int-fly}
            Set \emph{uncovNum} $:=$\emph{uncovNum}$+1$\;
            Add $\iota$ to \emph{unCovList}\;
            \If {uncovNum $>\rho$} {
                Set \emph{covered}$:=$ false\;
                \Break\;
            }
        }
    }
}
\For {each interaction $\iota \in $uncovList}{ \label{algline:second-stage}
    Add a row to $A$ that covers $\iota$\;
}
Output $A$\;

\caption{Na\"ive two-stage algorithm ($\TS\left<\mathrm{\rand,\naive}\right>$).}
\label{algo:two-stage-simple}
\end{algorithm}

This simple strategy  improves on the basic randomized strategy when $n$ is chosen judiciously. 
For example, when $t=6,\, k=54$ and $v=3$, Algorithm \ref{algo:rand-slj} constructs a covering array with $17,236$ rows. 
Figure \ref{fig:minima} plots an upper bound on the size of the covering array against the number $n$ of rows in the partial array.
The smallest covering array is obtained when $n=12,402$ which, when completed, yields a covering array with at most $13,162$ rows---a big improvement over Algorithm \ref{algo:rand-slj}.

\begin{figure}[htbp]
\begin{centering}
\includegraphics[bb=100bp 238bp 520bp 553bp,clip,scale=0.8]{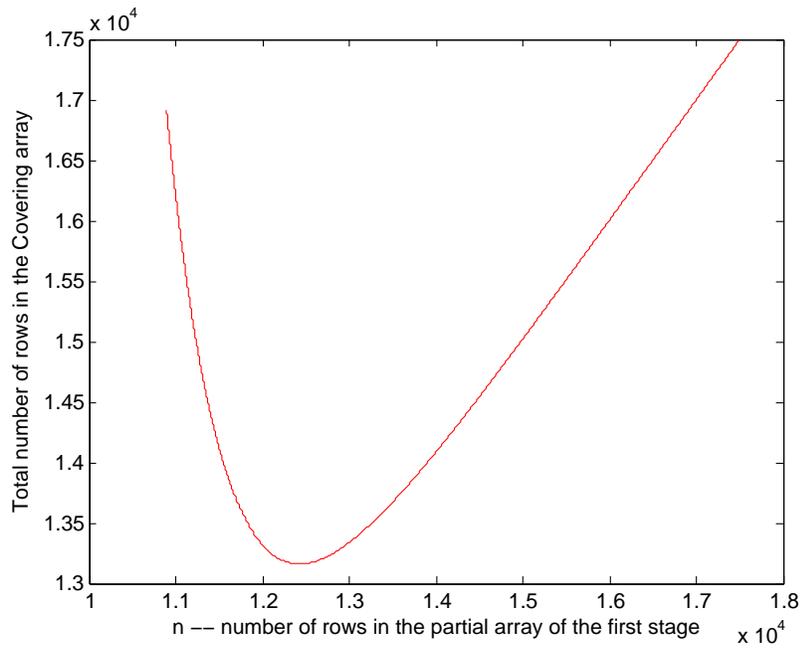}
\par\end{centering}

\caption{An upper bound on the size of the covering array against $n$, the size of the partial array constructed in the first stage when $t=6,k=54$, and $v=3$, with  one new row added per uncovered interaction in the second stage. The minimum size of $13,162$ is obtained when $n=12,402$. Algorithm \ref{algo:rand-slj} requires $17,236$ rows, and the currently best known covering array has $17,197$ rows.}
\label{fig:minima}
\end{figure}

A theorem from \cite{sarkar16} tells us the optimal value of $n$ in  general:
\begin{thm}
\label{thm:two-stage}{\rm \cite{sarkar16}} Let $t,\, k,\, v$ be integers with $k\ge t\ge2$,
and $v\ge2$. Then \[ \CAN(t,k,v)\le\frac{\log{k \choose t}+t\log v+\log\log\left(\frac{v^{t}}{v^{t}-1}\right)+1}{\log\left(\frac{v^{t}}{v^{t}-1}\right)}.\] 
\end{thm}

The  bound is obtained by setting  $n = \frac{\log{k \choose t}+t\log v+\log\log\left(\frac{v^{t}}{v^{t}-1}\right)}{\log\left(\frac{v^{t}}{v^{t}-1}\right)}$. 
The expected number of uncovered interactions  is exactly $\rho = 1/\log\left(\frac{v^{t}}{v^{t}-1}\right)$.

Figure \ref{fig:comp-slj-dslj-2stage} compares SLJ, discrete SLJ and two-stage bounds for $k \le 100$, when $t=6$ and $v=3$. 
The two-stage bound does not deteriorate in comparison to discrete SLJ bound as $k$ increases; it consistently takes only $307$-$309$ more rows. 
Thus when $k=12$ the two-stage bound requires only $~6\%$ more rows and when $k=100$ only $~2\%$ more rows than the discrete SLJ bound. 

\begin{figure}[htb]
\begin{centering}
\includegraphics[bb=100bp 238bp 520bp 553bp,clip,scale=0.8]{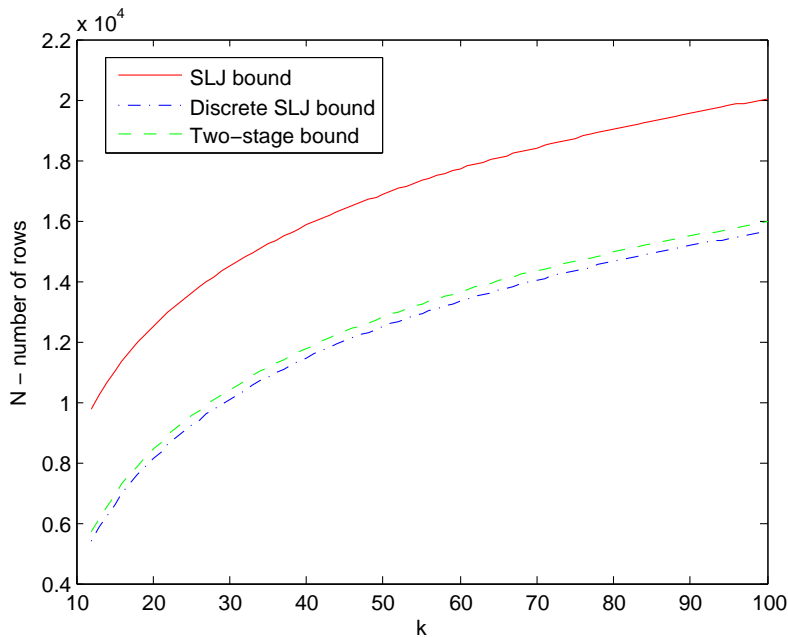}
\par\end{centering}

\caption{Comparison of covering array sizes obtained from SLJ bound, discrete SLJ bound and two-stage bound for $k \le 100$, when $t=6$ and $v=3$. In this range of $k$ values the two-stage bound requires $307$-$309$ more rows than the discrete SLJ bound, that is,  $2$-$6\%$ more rows.}
\label{fig:comp-slj-dslj-2stage}
\end{figure}

To ensure that the loop in line \ref{algline:check-Int-loop} of Algorithm \ref{algo:two-stage-simple} does not repeat too many times we need to know the probability with which a random $n \times k$ array leaves at most $\rho$ interactions uncovered. 
Using Chebyshev's inequality and the second moment method developed in  \cite[Chapter~4]{alon08}, we next show that in a random $n \times k$ array the number of uncovered interactions is almost always close to its expectation, i.e. ${k \choose t}v^t\left(1-\frac{1}{v^t}\right) ^n$. 
Substituting the value of $n$ from line \ref{algline:set-n},  this expected value is equal to $\mu$, as  in line \ref{algline:set-mu}. 
Therefore, the probability that a random $n \times k$ array covers the desired number of interactions is constant, and the expected number of times the loop in line \ref{algline:check-Int-loop} is repeated is also a constant (around $2$ in practice).

Because the theory of the second moment method is developed in considerable detail in \cite{alon08}, here we briefly mention the relevant concepts and results. 
Suppose that $X=\sum_{i=1}^m X_i$, where $X_i$ is the indicator random variable for event $A_i$ for $1\le i \le m$. 
For indices $i,j$, we write $i\sim j$ if $i\neq j$ and the events $A_i,A_j$ are not independent. 
Also suppose that $X_1,\ldots,X_m$ are \emph{symmetric}, i.e. for every $i\neq j$ there is a measure preserving mapping of the underlying probability space that sends event $A_i$ to event $A_j$. 
Define $\Delta^*=\sum_{j\sim i}\Pr\left[A_j|A_i\right]$. Then by  \cite[Corollary~4.3.4]{alon08}:
\begin{lem}
    \label{lem:2-moment}{\rm \cite{alon08}} If $\E[X]\rightarrow\infty$ and $\Delta^*=o(\E[X])$ then $X\sim \E[X]$ almost always.
\end{lem}

In our case, $A_i$ denotes the event that the $i$th interaction is \emph{not} covered in a $n \times k$ array where each entry is chosen independently and uniformly at random from a $v$-ary alphabet. 
Then $\Pr[X_i]=\left(1-\frac{1}{v^t}\right)^n$. 
Because there are ${k \choose t} v^t$ interactions in total, by linearity of expectation, $\E[X]={k \choose t}v^t\left(1-\frac{1}{v^t}\right)^n$, and $\E[X]\rightarrow\infty$ as $k\rightarrow \infty$.

Distinct events $A_i$ and $A_j$ are independent if the $i$th and $j$th interactions share no column. 
Therefore, the event $A_i$ is not independent of at most $t{k \choose t-1}$ other events $A_j$. 
So $\Delta^*=\sum_{j\sim i}\Pr\left[A_j|A_i\right] \le \sum_{j\sim i}1\le t{k \choose t-1}= o(\E[X])$ when $v$ and $t$ are constants. 
By Lemma \ref{lem:2-moment},  the number of uncovered interactions in a random $n \times k$ array is close to the expected number of uncovered interactions. 
This guarantees that Algorithm \ref{algo:two-stage-simple} is an efficient randomized algorithm for constructing covering arrays with a number of rows upper bounded by Theorem \ref{thm:two-stage}.

In keeping  with the general two-stage framework, Algorithm \ref{algo:two-stage-simple} does not store the coverage status of each interaction. 
We only need  store the interactions that are uncovered in $A$, of which there are at most $\rho = \frac{1}{\log\left(\frac{v^{t}}{v^{t}-1}\right)} \approx v^t$. 
This quantity depends only on $v$ and $t$ and is independent of $k$, so is effectively a constant that is much smaller than ${k \choose t}  v^t$, the storage requirement for the density algorithm. 
Hence the algorithm  can be applied to a  higher range of $k$ values. 

Although Theorem \ref{thm:godbole} provides asymptotically tighter bounds than Theorem \ref{thm:two-stage}, in a range of $k$ values that are relevant for practical application, Theorem \ref{thm:two-stage} provides better results.
Figure \ref{fig:comp-lll-2stage-best} compares the bounds on $\CAN(t,k,v)$   with the currently best known results. 

\begin{figure}[htb]
    \centering
    \begin{subfigure}[b]{0.4\textwidth}
        \includegraphics[bb=100bp 238bp 520bp 553bp,clip,scale=0.5]{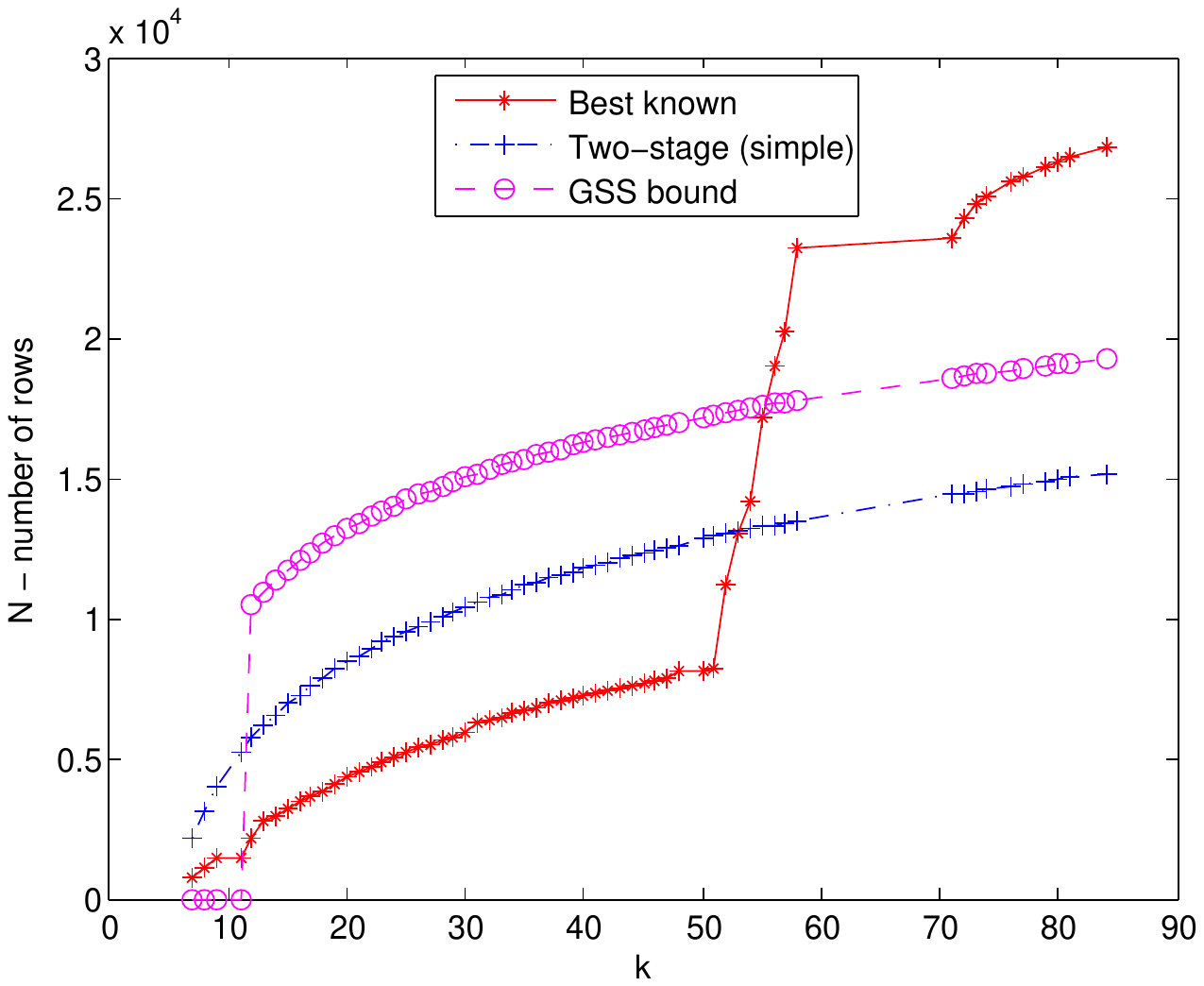}
        \caption{$t=6,\,v=3$}
    \end{subfigure}
    \quad 
    \begin{subfigure}[b]{0.4\textwidth}
        \includegraphics[bb=100bp 238bp 520bp 553bp,clip,scale=0.5]{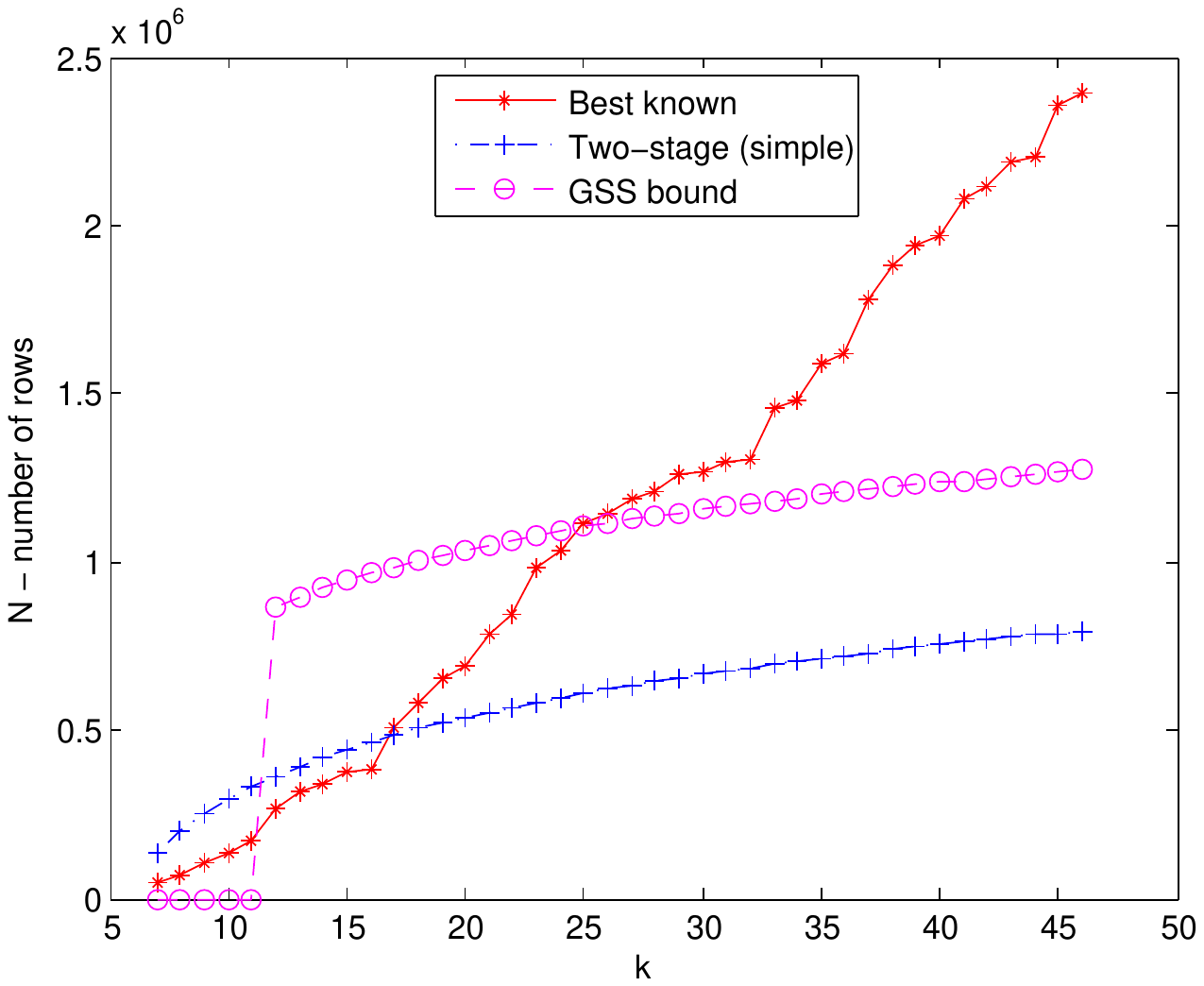}
        \caption{$t=6,\,v=6$}
    \end{subfigure}
    \caption{Comparison of GSS bound and two-stage bound with the currently best known results}\label{fig:comp-lll-2stage-best}
\end{figure}

\subsection{The density algorithm in the second stage ($\TS\left<\mathrm{\rand,\den}\right>$)}\label{subsec:density}
Next we  apply the density algorithm in the second stage. 
Figure \ref{fig:minima-comp} plots an upper bound on the size of the covering array against the size of the partial array constructed in the first stage when the density algorithm is used in the second stage, and compares it with $\TS\left<\mathrm{\rand,\naive}\right>$. 
The size of the covering array decreases as $n$ decreases. 
This is expected because with smaller partial arrays, more interactions remain for the second stage to be covered  by the density algorithm. 
In fact if we cover all the interactions using the density algorithm (as when $n=0$) we would get an even smaller covering array. 
However, our motivation was precisely to avoid doing  that. 
Therefore, we need  a "cut-off" for the first stage. 

\begin{figure}[htbp]
\begin{centering}
    \includegraphics[bb=100bp 238bp 520bp 553bp,clip,scale=0.8]{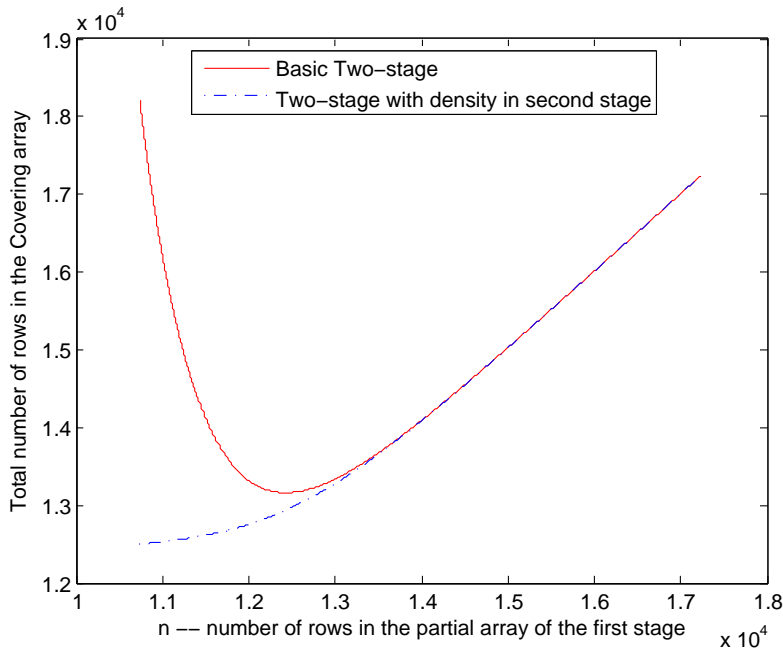}
    \par\end{centering}
    
    \caption{Comparison of covering array sizes  from two-stage algorithms with \den\ and \naive\ in the second stage. 
  With \den\ there is no minimum point in the curve; the size of the covering array keeps decreasing as we leave more  uncovered interactions for  the second stage.}
    \label{fig:minima-comp}
\end{figure}

We are presented with a trade-off. 
If we construct a smaller partial array in the first stage,  we obtain a smaller covering array overall. 
But we then pay for more storage and computation time for  the second stage. 
To appreciate the nature of this trade-off,  look at Figure \ref{fig:emp-cutoff}, 
which plots an upper bound on the covering array size and the number of uncovered interactions in the first stage against $n$. 
The improvement in the covering array size plateaus after a certain point. 
The three horizontal lines indicate $\rho\,(\approx v^t)$, $2\rho$ and $3\rho$ uncovered interactions in the first stage. 
(In the na\"ive method of Section \ref{subsec:simple}, the partial array after the first stage leaves at most $\rho$ uncovered interactions.)
In Figure \ref{fig:emp-cutoff} the final covering array size appears to plateau when the number of uncovered interactions left by the first stage is around $2\rho$. 
After that we see diminishing returns --- the density algorithm needs to cover  more interactions in return for a smaller improvement in the covering array size. 

Let $r$ be the maximum number of interactions allowed to remain uncovered after the first stage. 
Then $r$ can be specified in the  two-stage algorithm. 
To accommodate this, we  denote by  $\TS\left<A,B;r\right>$ the two-stage algorithm where $A$ is the first stage strategy, $B$ is the second stage strategy,  and $r$ is the maximum number of uncovered interactions after the first stage. 
For example, $\TS\left<\mathrm{\rand,\den};2\rho\right>$ applies the basic randomized algorithm in the first stage to cover all but at most $2\rho$ interactions, and the density algorithm to cover the remaining interactions in the second stage.

\begin{figure}[htbp]
\begin{centering}
    \includegraphics[bb=100bp 238bp 520bp 553bp,clip,scale=0.8]{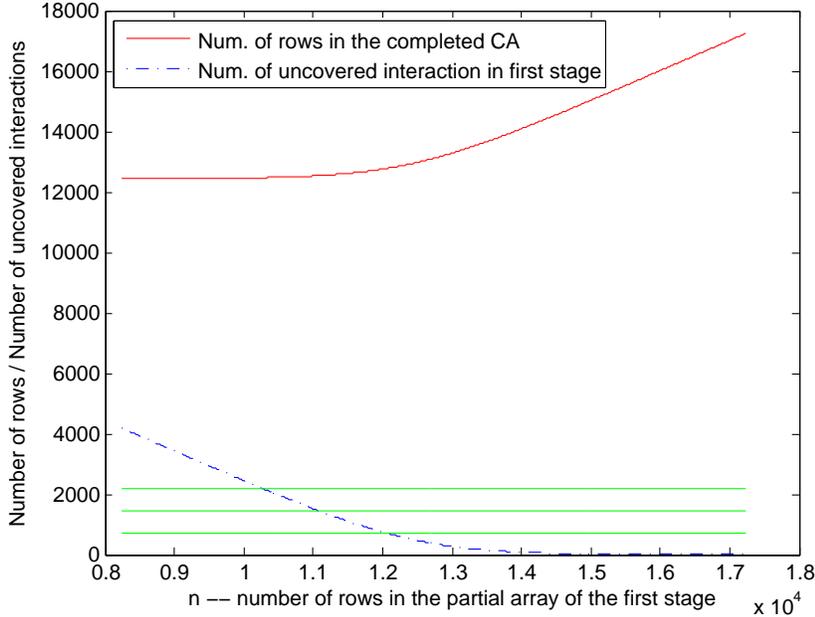}
    \par\end{centering}
    
    \caption{Final covering array size against the number of uncovered interactions after the first stage. As the size $n$ of the partial array  decreases, the number of uncovered interactions in the first stage increases. \den\ is used in the second stage. From bottom to top, the  green lines denote $\rho$, $2\rho$, and $3\rho$  uncovered interactions.}
    \label{fig:emp-cutoff}
\end{figure}

\subsection{Coloring in the second stage ($\TS\left<\mathrm{\rand,\color}\right>$ and $\TS\left<\mathrm{\rand,\greedy}\right>$)}\label{subsec:coloring}

Now we describe strategies using graph coloring   in the second stage.
Construct a graph $G=(V,E)$, the \emph{incompatibility graph}, in which $V$ is the set of uncovered
interactions, and there is an edge between two interactions exactly when
they share a column in which they have different symbols.
A single row can cover a set of interactions if and only if it forms an independent set in $G$.  
Hence the minimum number of rows required to cover all interactions of $G$ is exactly its chromatic number $\chi(G)$, the minimum number of colors in a proper coloring of $G$.
Graph coloring is an NP-hard problem, so we employ heuristics to bound the chromatic number.
Moreover, $G$ only has vertices for the \emph{uncovered} interactions after the first stage, so is size is small relative to the total number of interactions.

The expected number of edges in the  incompatibility graph after choosing $n$ rows uniformly at random is  $\gamma=\left(\frac{1}{2}\right) {k \choose t} v^{t}  \sum_{i=1}^{t}{t \choose i}{k-t \choose t-i}(v^{t}-v^{t-i}) \left(1-\frac{1}{v^{t}}\right)^{n} \left(1-\frac{1}{(v^{t}-v^{t-i})}\right)^{n}$.
Using the elementary upper bound on the chromatic number $\chi\le\frac{1}{2}+\sqrt{2m+\frac{1}{4}}$, where $m$ is the number of edges \cite[Chapter~5.2]{diestel10}, 
we can surely cover the remaining interactions with at most $\frac{1}{2}+\sqrt{2m+\frac{1}{4}}$ rows. 

The actual number of edges $m$ that remain after the first stage is a random variable with mean $\gamma$.
In principle, the first stage could be repeatedly applied until $m \leq \gamma$, so we call $m=\gamma$ the  \emph{optimistic estimate}.
To ensure that the first stage is expected to be run a small constant number of times, we increase the estimate.
With probability more than $1/2$ the  incompatibility graph has  $m\le2\gamma$ edges, so $m=2\gamma$ is the \emph{conservative estimate}.

For $t=6,\, k=56,$ and $v=3$, 
Figure \ref{fig:N-vs-n} shows the effect on the minimum number of rows when the bound on the chromatic number in the second stage is used, for the conservative or optimistic estimates.  The Na\"ive method is plotted for comparison.
Better coloring bounds shift the minima leftward, reducing the number of rows produced in both stages.

\begin{figure}[htbp]
\begin{centering}
    \includegraphics[bb=100bp 238bp 520bp 553bp,clip,scale=0.8]{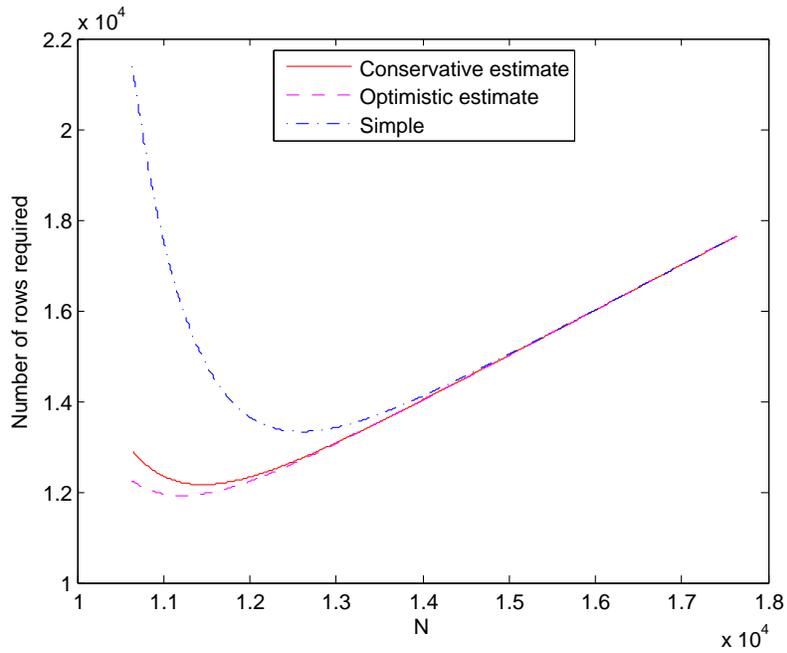}
    \par\end{centering}
    
    \caption{Size of the partial array vs. size of the complete CA. $t=6,\, k=56,\, v=3$. Stein-Lovász-Johnson bound requires $17,403$ rows, discrete Stein-Lovász-Johnson bound requires $13,021$ rows. Simple estimate for the two stage algorithm is $13,328$ rows, conservative estimate assuming $m=2\gamma$ is $12,159$ rows, and optimistic estimate assuming $m=\gamma$ is $11,919$ rows. Even the conservative estimate beats the discrete Stein-Lovász-Johnson bound.}
    \label{fig:N-vs-n}
\end{figure}

Thus far we have considered bounds on the chromatic number.  
Better estimation of $\chi(G)$ is complicated by the fact that we do not have much information about the structure of $G$ until the first stage is run. 
In practice, however, $G$ is known after the first stage and hence an algorithmic method to bound its chromatic number can be applied.  
Because the number of vertices in $G$ equals the number of uncovered interactions after the first stage, we encounter the same trade-off between time and storage, and final array size, as seen earlier for density.  
Hence we again parametrize  by the expected number of uncovered interactions in the first stage.

We employ two different greedy algorithms to color the incompatibility graph. 
In method  \color\   we first construct the  incompatibility graph $G$ after  the first stage. 
Then we apply the commonly used \emph{smallest last order} heuristic to order the vertices for greedy coloring: At each stage, find a vertex $v_i$ of minimum degree in $G_i$, order the vertices of $G_i-v_i$, and then place $v_i$ at the end. 
More precisely, we order the vertices of $G$ as $v_1,v_2, \ldots, v_n$, such that $v_i$ is a vertex of minimum degree in $G_i$, where $G_i = G - \{v_{i+1}, \ldots, v_n \}$. 
A graph is  \emph{$d$-degenerate} if all of its subgraphs have a vertex with degree at most $d$. 
When $G$ is $d$-degenerate but not $(d-1)$-degenerate, the \emph{Coloring number} $\col(G)$ is $d+1$. 
If we  then greedily color the vertices with the first available color, at most $\col(G)$ colors are used.

In  method  \greedy\  we employ an on-line, greedy approach that colors the interactions as they are discovered in the first stage.
In this way, the incompatibility graph is never constructed.
We instead maintain a set of rows.
Some entries in  rows are fixed to a specific value;  others are flexible to take on any value.
Whenever a new interaction is found to be uncovered in the first stage, we check if any of the rows  is compatible with this interaction.
If  such a row is found then  entries in the row are fixed so that the row now covers the interaction.
If no such row exists,  a new row with exactly $t$ fixed entries corresponding to the interaction is added to the set of rows.
This method is much faster than method \color\  in practice.

\subsection{Using group action}\label{subsec:group}

Covering arrays  that are invariant under the action of a permutation group on their symbols can be easier to construct and are often smaller \cite{ColCECA}.
Direct and computational constructions using group actions are explored in \cite{cck,MeagherS}.
Sarkar et al. \cite{sarkar16} establish the asymptotically tightest known bounds on $\CAN(t,k,v)$ using group actions. 
In this section we explore the implications of group actions on  two-stage algorithms.

Let $\Gamma$ be a permutation group on the set of symbols. 
The action of this group partitions the set of  $t$-way interactions into  orbits. 
We construct an array $A$ such that for every orbit, at least one row covers an interaction from that orbit. 
Then we develop the rows of $A$ over $\Gamma$ to obtain a covering array  that is invariant under the action of $\Gamma$.
Effort then focuses on covering all the  orbits of $t$-way interactions, instead of the individual interactions.

If $\Gamma$ acts sharply transitively on the set of symbols (for example, if $\Gamma$ is a cyclic group of order $v$) then the action of $\Gamma$ partitions ${k \choose t}  v^t$ interactions into ${k \choose t}  v^{t-1}$ orbits of length $v$ each. 
Following the lines of the proof of Theorem \ref{thm:two-stage}, there exists an $n \times k$ array with $n =\frac{\log{k \choose t}+(t-1)\log v+\log\log\left(\frac{v^{t-1}}{v^{t-1}-1}\right)+1}{\log\left(\frac{v^{t-1}}{v^{t-1}-1}\right)}$ that covers at least one interaction from each orbit. 
Therefore, 
\begin{equation}\label{eq:cyclic}
\CAN(t,k,v) \le v\frac{\log{k \choose t}+(t-1)\log v+\log\log\left(\frac{v^{t-1}}{v^{t-1}-1}\right)+1}{\log\left(\frac{v^{t-1}}{v^{t-1}-1}\right)}.
\end{equation}

\begin{figure}[htbp]
\begin{centering}
    \includegraphics[bb=100bp 238bp 520bp 553bp,clip,scale=0.8]{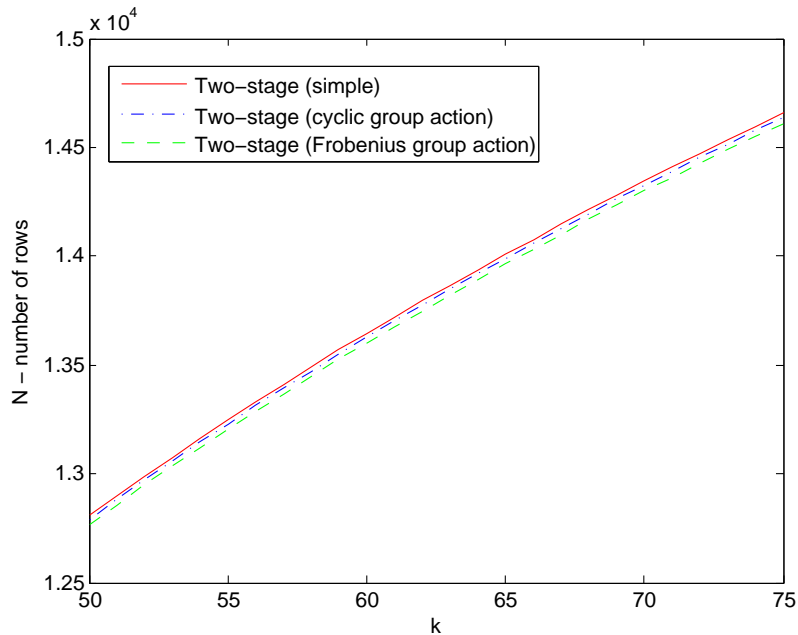}
    \par\end{centering}
    
    \caption{Comparison of the simple two-stage bound with the cyclic and the Frobenius   two-stage bounds. $t=6,\,v=3$ and $50\le k \le 75$. Group action reduces the required number of rows slightly.}
    \label{fig:comp-2stage}
\end{figure}

Similarly, we can  employ  a Frobenius group.
When $v$ is a prime power, the \emph{Frobenius group} is the group of permutations of $\mathbb{F}_v$ of the form $\{x \mapsto ax+b\,:\,a,b\in \mathbb{F}_v,\,a\neq0\}$.
The action of the Frobenius group partitions the set of $t$-tuples on $v$ symbols  into $\frac{v^{t-1}-1}{v-1}$ orbits of length $v(v-1)$ (full orbits) each and $1$ orbit of length $v$ (a short orbit).
The short orbit consists of tuples of the form $(x_1,\ldots,x_t)$ where $x_1=\ldots=x_t$.
Therefore, we can obtain a covering array by first constructing an array that covers all the full orbits, and then developing all the rows over the Frobenius group and adding $v$ constant rows.
Using the two stage strategy in conjunction with the Frobenius group action we obtain:

\begin{equation}\label{eq:frobenius}
\CAN(t,k,v) \le v(v-1)\frac{\log{k \choose t}+ \log\left(\frac{v^{t-1}-1}{v-1}\right) +\log\log\left(\frac{v^{t-1}}{v^{t-1}-v+1}\right)+1}{\log\left(\frac{v^{t-1}}{v^{t-1}-v+1}\right)} + v.
\end{equation}

Figure \ref{fig:comp-2stage} compares the simple two-stage bound with the cyclic and  Frobenius  two-stage bounds. 
For $t=6,\,v=3$ and $12\le k \le 100$, the cyclic bound requires $7$-$21$ (on average $16$) fewer rows than the simple bound.
In the same range the Frobenius bound requires $17-51$  (on average $40$) fewer rows.

Group action can be applied in other methods for the second stage as well. 
Colbourn \cite{ColCECA} incorporates group action into the density algorithm, allowing us to apply method \den\  in the second stage.

\greedy\  extends  easily to use group action, as we do not  construct an explicit incompatibility graph.
Whenever we fix entries in a row to cover an uncovered orbit,  we commit to a specific orbit representative. 

However, applying group action to the incompatibility graph coloring for \color\ is more complicated.
We need to modify the definition of the incompatibility graph for two reasons.
First the vertices no longer  represent  uncovered interactions, but rather uncovered orbits of interaction.
Secondly, and perhaps more importantly, pairwise compatibility between every two orbits in a set no longer implies mutual compatibility among all orbits in  the set.  

One approach is to form a vertex for each uncovered orbit, placing an edge between two when they share a column.  
Rather than the usual coloring, however, one asks for a partition of the vertex set into classes so that every class induces an acyclic subgraph.  
Problems of this type are \emph{generalized graph coloring} problems \cite{Brown96}.
Within each  class of such a vertex partition, consistent representatives of each orbit can be selected to form a row; when a cycle is present, this may not be possible.
Unfortunately, heuristics for solving these types of problems appear to be weak, so we adopt a second approach.
As we build the incompatibility graph, we commit to specific orbit representatives. 
When a vertex for an uncovered orbit is added, we check its compatibility with the orbit representatives chosen for the orbits already handled with which it shares columns; 
we commit to an orbit representative and add edges to those with which it is now incompatible.
Once completed, we have a (standard) coloring problem for the resulting graph.

Because group action can be applied using each of the methods for the two stages,  we extend our naming  to $\TS\left<A,B;r,\Gamma\right>$, where $\Gamma$ can be  \T\   (i.e. no group action), \C, or \F.

\section{Computational results}\label{sec:results}

Figure \ref{fig:comp-lll-2stage-best} indicates that even a simple two-stage bound can improve on best known covering array numbers. 
Therefore we investigate the actual performance of our two-stage algorithms for covering arrays of strength $5$ and $6$. 

First we present results for $t=6$, when $v\in\{3,4,5,6\}$ and no group action is assumed.
Table \ref{tab:t-1-6} shows the results for different $v$ values.
In each case we select the range of $k$ values where the two-stage bound predicts smaller covering arrays than the previously known best ones, setting the maximum number of uncovered interactions as $\rho = 1/\log\left(\frac{v^{t}}{v^{t}-1}\right) \approx v^t$.
For each value of $k$ we construct a single partial array and then run the different second stage algorithms on it consecutively. 
In this way all the second stage algorithms cover the same set of uncovered interactions.

The column  \textbf{tab} lists the best known $\CAN(t,k,v)$ upper bounds from \cite{catables}.
The column  \textbf{bound} shows the upper bounds obtained from the two-stage bound (\ref{thm:two-stage}).
The columns \textbf{na\"ive}, \textbf{greedy}, \textbf{col} and  \textbf{den} show results obtained from running the $\TS\left<\rand,\naive;\rho,\T\right>$, $\TS\left<\rand, \greedy;\rho,\T\right>$, $\TS\left<\rand, \color;\rho,\T\right>$ and $\TS\left<\rand,\den;\rho,\T\right>$ algorithms, respectively.

The na\"ive method always finds a covering array that is smaller than the two-stage bound. 
This happens because we repeat the first stage of Algorithm \ref{algo:two-stage-simple} until the array has fewer than $v^t$ uncovered interactions.
(If the first stage were not repeated,  the algorithm  still produce covering arrays that are not too far from the bound.)
For $v=3$ \greedy\ and \den\  have comparable performance.
Method \color\  produces covering arrays that are smaller.
However, for $v\in \{4,5,6\}$ \den\ and \color\ are competitive.

\begin{table}[htbp]
\begin{centering}
\begin{tabular}{c|c|c|c|c|c|c}
\hline 
$k$ & tab & bound & na\"ive & greedy & col & den\tabularnewline
\hline 
\multicolumn{7}{c}{$t=6,v=3$}\\
\hline 
\hline 
53 & 13021 & 13076 & 13056 & 12421 & 12415 & 12423\tabularnewline
\hline 
54 & 14155 & 13162 & 13160 & 12510 & 12503 & 12512\tabularnewline
\hline 
55 & 17161 & 13246 & 13192 & 12590 & 12581 & 12591\tabularnewline
\hline 
56 & 19033 & 13329 & 13304 & 12671 & 12665 & 12674\tabularnewline
\hline 
57 & 20185 & 13410 & 13395 & 12752 & 12748 & 12757\tabularnewline
\hline 
\hline 
\multicolumn{7}{c}{$t=6,v=4$}\\
\hline 
\hline 
39 & 68314 & 65520 & 65452 & 61913 & 61862 & 61886\tabularnewline
\hline 
40 & 71386 & 66186 & 66125 & 62573 & 62826 & 62835\tabularnewline
\hline 
41 & 86554 & 66834 & 66740 & 63209 & 63160 & 63186\tabularnewline
\hline 
42 & 94042 & 67465 & 67408 & 63819 & 64077 & 64082\tabularnewline
\hline 
43 & 99994 & 68081 & 68064 & 64438 & 64935 & 64907\tabularnewline
\hline 
44 & 104794 & 68681 & 68556 & 65021 & 65739 & 65703\tabularnewline
\hline 
\hline 
\multicolumn{7}{c}{$t=6,v=5$}\\
\hline 
\hline 
31 & 233945 & 226700 & 226503 & 213244 & 212942 & 212940\tabularnewline
\hline 
32 & 258845 & 229950 & 229829 & 216444 & 217479 & 217326\tabularnewline
\hline 
33 & 281345 & 233080 & 232929 & 219514 & 219215 & 219241\tabularnewline
\hline 
34 & 293845 & 236120 & 235933 & 222516 & 222242 & 222244\tabularnewline
\hline 
35 & 306345 & 239050 & 238981 & 225410 & 226379 & 226270\tabularnewline
\hline 
36 & 356045 & 241900 & 241831 & 228205 & 230202 & 229942\tabularnewline
\hline 
\hline 
\multicolumn{7}{c}{$t=6,v=6$}\\
\hline 
\hline 
17 & 506713 & 486310 & 486302 & 449950 & 448922 & 447864\tabularnewline
\hline 
18 & 583823 & 505230 & 505197 & 468449 & 467206 & 466438\tabularnewline
\hline 
19 & 653756 & 522940 & 522596 & 485694 & 484434 & 483820\tabularnewline
\hline 
20 & 694048 & 539580 & 539532 & 502023 & 500788 & 500194\tabularnewline
\hline 
21 & 783784 & 555280 & 555254 & 517346 & 516083 & 515584\tabularnewline
\hline 
22 & 844834 & 570130 & 569934 & 531910 & 530728 & 530242\tabularnewline
\hline 
23 & 985702 & 584240 & 584194 & 545763 & 544547 & 548307\tabularnewline
\hline 
24 & 1035310 & 597660 & 597152 & 558898 & 557917 & 557316\tabularnewline
\hline 
25 & 1112436 & 610460 & 610389 & 571389 & 570316 & 569911\tabularnewline
\hline 
26 & 1146173 & 622700 & 622589 & 583473 & 582333 & 582028\tabularnewline
\hline 
27 & 1184697 & 634430 & 634139 & 594933 & 593857 & 593546\tabularnewline
\hline 
\hline
\end{tabular}
\par\end{centering}

\caption{Comparison of different $\TS\left<\rand,-;\rho,\T\right>$ algorithms.}
\label{tab:t-1-6}
\end{table}

Table \ref{tab:t-2-6}  shows the results obtained by the different second stage algorithms when the maximum number of uncovered interactions in the first stage is set to $2\rho$ and $3\rho$ respectively.
When more interactions are covered in the second stage, we obtain smaller arrays as expected.
However, the improvement in size does not approach $50\%$.
There is no clear winner.

\begin{table}[htbp]
\begin{centering}
\begin{tabular}{c|c|c|c||c|c|c}
\hline 
$k$ & \multicolumn{1}{c}{} & \multicolumn{1}{c}{$2\rho$} &  & \multicolumn{1}{c}{} & \multicolumn{1}{c}{$3\rho$} & \tabularnewline
\cline{2-7} 
 & greedy & col & den & greedy & col & den\tabularnewline
\hline 
\hline 
\multicolumn{7}{c}{$t=6,v=3$}\\
\hline 
\hline 
53 & 11968 & 11958 & 11968 & 11716 & 11705 & 11708\tabularnewline
\hline 
54 & 12135 & 12126 & 12050 & 11804 & 11787 & 11790\tabularnewline
\hline 
55 & 12286 & 12129 & 12131 & 11877 & 11875 & 11872\tabularnewline
\hline 
56 & 12429 & 12204 & 12218 & 11961 & 12055 & 11950\tabularnewline
\hline 
57 & 12562 & 12290 & 12296 & 12044 & 12211 & 12034\tabularnewline
\hline 
\multicolumn{7}{c}{$t=6,v=4$}\\
\hline 
\hline 
39 & 59433 & 59323 & 59326 & 58095 & 57951 & 57888\tabularnewline
\hline 
40 & 60090 & 60479 & 59976 & 58742 & 58583 & 58544\tabularnewline
\hline 
41 & 60715 & 61527 & 60615 & 59369 & 59867 & 59187\tabularnewline
\hline 
42 & 61330 & 62488 & 61242 & 59974 & 61000 & 59796\tabularnewline
\hline 
43 & 61936 & 61839 & 61836 & 60575 & 60407 & 60393\tabularnewline
\hline 
44 & 62530 & 62899 & 62428 & 61158 & 61004 & 60978\tabularnewline
\hline 
\multicolumn{7}{c}{$t=6,v=5$}\\
\hline 
\hline 
31 & 204105 & 203500 & 203302 & 199230 & 198361 & 197889\tabularnewline
\hline 
32 & 207243 & 206659 & 206440 & 202342 & 201490 & 201068\tabularnewline
\hline 
33 & 210308 & 209716 & 209554 & 205386 & 204548 & 204107\tabularnewline
\hline 
34 & 213267 & 212675 & 212508 & 208285 & - & 207060\tabularnewline
\hline 
35 & 216082 & 215521 & 215389 & 211118 & - & 209936\tabularnewline
\hline 
36 & 218884 & 218314 & 218172 & 213872 & - & 212707\tabularnewline
\hline 
\multicolumn{7}{c}{$t=6,v=6$}\\
\hline 
\hline 
17 & 425053 & - & 420333 & 412275 & - & 405093\tabularnewline
\hline 
18 & 443236 & - & 438754 & 430402 & - & 423493\tabularnewline
\hline 
19 & 460315 & - & 455941 & 447198 & - & 440532\tabularnewline
\hline 
20 & 476456 & - & 472198 & 463071 & - & 456725\tabularnewline
\hline 
21 & 491570 & - & 487501 & 478269 & - & 471946\tabularnewline
\hline 
22 & 505966 & - & 502009 & 492425 & - & 486306\tabularnewline
\hline 
23 & 519611 & - & 515774 & 505980 & - & 500038\tabularnewline
\hline 
24 & 532612 & - & 528868 & 518746 & - & 513047\tabularnewline
\hline 
25 & 544967 & - & 541353 & 531042 & - & 525536\tabularnewline
\hline 
26 & 556821 & - & 553377 & 542788 & - & 537418\tabularnewline
\hline 
27 & 568135 & - & 564827 & 554052 & - & 548781\tabularnewline
\hline 
\end{tabular}
\par\end{centering}

\caption{Comparison of  $\TS\left<\rand,-;2\rho,\T\right>$  and $\TS\left<\rand,-;3\rho,\T\right>$ algorithms.}
\label{tab:t-2-6}
\end{table}

Next we investigate the covering arrays that are invariant under the action of a cyclic group.
In Table \ref{tab:c-1-6} the column  \textbf{bound} shows the upper bounds from Equation (\ref{eq:cyclic}).
The columns \textbf{na\"ive}, \textbf{greedy}, \textbf{col} and  \textbf{den} show results obtained from running  $\TS\left<\rand,\naive;\rho,\C\right>$, $\TS\left<\rand, \greedy;\rho,\C\right>$, $\TS\left<\rand, \color;\rho,\C\right>$ and $\TS\left<\rand,\den;\rho,\C\right>$, respectively.

\begin{table}[htbp]
\begin{centering}
\begin{tabular}{c|c|c|c|c|c|c}
\hline 
$k$ & tab & bound & na\"ive & greedy & col & den\tabularnewline
\hline 
\hline 
\multicolumn{7}{c}{$t=6,v=3$}\\
\hline 
\hline 
53 & 13021 & 13059 & 13053 & 12405 & 12405 & 12411\tabularnewline
\hline 
54 & 14155 & 13145 & 13119 & 12489 & 12543 & 12546\tabularnewline
\hline 
55 & 17161 & 13229 & 13209 & 12573 & 12663 & 12663\tabularnewline
\hline 
56 & 19033 & 13312 & 13284 & 12660 & 12651 & 12663\tabularnewline
\hline 
57 & 20185 & 13393 & 13368 & 12744 & 12744 & 12750\tabularnewline
\hline 
\hline 
\multicolumn{7}{c}{$t=6,v=4$}\\
\hline 
\hline 
$k$ & tab & bound & na\"ive & greedy & col & den\tabularnewline
\hline 
\hline 
39 & 68314 & 65498 & 65452 & 61896 & 61860 & 61864\tabularnewline
\hline 
40 & 71386 & 66163 & 66080 & 62516 & 62820 & 62784\tabularnewline
\hline 
41 & 86554 & 66811 & 66740 & 63184 & 63144 & 63152\tabularnewline
\hline 
42 & 94042 & 67442 & 67408 & 63800 & 63780 & 63784\tabularnewline
\hline 
43 & 99994 & 68057 & 68032 & 64408 & 64692 & 64680\tabularnewline
\hline 
44 & 104794 & 68658 & 68556 & 64988 & 64964 & 64976\tabularnewline
\hline 
\hline 
\multicolumn{7}{c}{$t=6,v=5$}\\
\hline 
\hline 
31 & 226000 & 226680 & 226000 & 213165 & 212945 & 212890\tabularnewline
\hline 
32 & 244715 & 229920 & 229695 & 216440 & 217585 & 217270\tabularnewline
\hline 
33 & 263145 & 233050 & 233015 & 219450 & 221770 & 221290\tabularnewline
\hline 
34 & 235835 & 236090 & 235835 & 222450 & 222300 & 222210\tabularnewline
\hline 
35 & 238705 & 239020 & 238705 & 225330 & 225130 & 225120\tabularnewline
\hline 
36 & 256935 & 241870 & 241470 & 228140 & 229235 & 229020\tabularnewline
\hline 
\hline 
\multicolumn{7}{c}{$t=6,v=6$}\\
\hline 
\hline 
17 & 506713 & 486290 & 485616 & 449778 & 448530 & 447732\tabularnewline
\hline 
18 & 583823 & 505210 & 504546 & 468156 & 467232 & 466326\tabularnewline
\hline 
19 & 653756 & 522910 & 522258 & 485586 & 490488 & 488454\tabularnewline
\hline 
20 & 694048 & 539550 & 539280 & 501972 & 500880 & 500172\tabularnewline
\hline 
21 & 783784 & 555250 & 554082 & 517236 & 521730 & 519966\tabularnewline
\hline 
22 & 844834 & 570110 & 569706 & 531852 & 530832 & 530178\tabularnewline
\hline 
23 & 985702 & 584210 & 583716 & 545562 & 549660 & 548196\tabularnewline
\hline 
24 & 1035310 & 597630 & 597378 & 558888 & 557790 & 557280\tabularnewline
\hline 
25 & 1112436 & 610430 & 610026 & 571380 & 575010 & 573882\tabularnewline
\hline 
26 & 1146173 & 622670 & 622290 & 583320 & 582546 & 582030\tabularnewline
\hline 
27 & 1184697 & 624400 & 633294 & 594786 & 598620 & 597246\tabularnewline
\hline 
\hline
\end{tabular}
\par\end{centering}

\caption{Comparison of  $\TS\left<\rand,-;\rho,\C\right>$algorithms.}
\label{tab:c-1-6}
\end{table}

Table \ref{tab:c-2-6} presents results for cyclic group action based algorithms when the number of maximum uncovered interactions in the first stage is set to $2\rho$ and $3\rho$ respectively.

\begin{table}[htbp]
\begin{centering}
\begin{tabular}{c|c|c|c||c|c|c}
\hline 
$k$ & \multicolumn{1}{c}{} & \multicolumn{1}{c}{$2\rho$} &  & \multicolumn{1}{c}{} & \multicolumn{1}{c}{$3\rho$} & \tabularnewline
\cline{2-7} 
 & greedy & col & den & greedy & col & den\tabularnewline
\hline 
\hline 
\multicolumn{7}{c}{$t=6,v=3$}\\
\hline 
\hline 
53 & 11958 & 11955 & 11958 & 11700 & 11691 & 11694\tabularnewline
\hline 
54 & 12039 & 12027 & 12036 & 11790 & 11874 & 11868\tabularnewline
\hline 
55 & 12120 & 12183 & 12195 & 11862 & 12057 & 12027\tabularnewline
\hline 
56 & 12204 & 12342 & 12324 & 11949 & 11937 & 11943\tabularnewline
\hline 
57 & 12276 & 12474 & 12450 & 12027 & 12021 & 12024\tabularnewline
\hline 
\hline 
\multicolumn{7}{c}{$t=6,v=4$}\\
\hline 
\hline 
39 & 59412 & 59336 & 59304 & 58076 & 57976 & 57864\tabularnewline
\hline 
40 & 60040 & 59996 & 59964 & 58716 & 58616 & 58520\tabularnewline
\hline 
41 & 60700 & 61156 & 61032 & 59356 & 59252 & 59160\tabularnewline
\hline 
42 & 61320 & 62196 & 61976 & 59932 & 59840 & 59760\tabularnewline
\hline 
43 & 61908 & 63192 & 62852 & 60568 & 61124 & 60904\tabularnewline
\hline 
44 & 62512 & 64096 & 63672 & 61152 & 61048 & 60988\tabularnewline
\hline 
\hline 
\multicolumn{7}{c}{$t=6,v=5$}\\
\hline 
\hline 
31 & 204060 & 203650 & 203265 & 199180 & 198455 & 197870\tabularnewline
\hline 
32 & 207165 & 209110 & 208225 & 202255 & 204495 & 203250\tabularnewline
\hline 
33 & 207165 & 209865 & 209540 & 205380 & 204720 & 204080\tabularnewline
\hline 
34 & 213225 & 212830 & 212510 & 208225 & 207790 & 207025\tabularnewline
\hline 
35 & 216050 & 217795 & 217070 & 211080 & 213425 & 212040\tabularnewline
\hline 
36 & 218835 & 218480 & 218155 & 213770 & 213185 & 212695\tabularnewline
\hline 
\hline 
\multicolumn{7}{c}{$t=6,v=6$}\\
\hline 
\hline 
17 & 424842 & 422736 & 420252 & 411954 & 409158 & 405018\tabularnewline
\hline 
18 & 443118 & 440922 & 438762 & 430506 & 427638 & 423468\tabularnewline
\hline 
19 & 460014 & 457944 & 455994 & 447186 & 456468 & 449148\tabularnewline
\hline 
20 & 476328 & 474252 & 472158 & 463062 & 460164 & 456630\tabularnewline
\hline 
21 & 491514 & 489270 & 487500 & 478038 & 486180 & 479970\tabularnewline
\hline 
22 & 505884 & 503580 & 501852 & 492372 & 489336 & 486264\tabularnewline
\hline 
23 & 519498 & 517458 & 515718 & 505824 & 502806 & 500040\tabularnewline
\hline 
24 & 532368 & 530340 & 528828 & 518700 & 515754 & 512940\tabularnewline
\hline 
25 & 544842 & 542688 & 541332 & 530754 & 538056 & 532662\tabularnewline
\hline 
26 & 543684 & 543684 & 543684 & 542664 & 539922 & 537396\tabularnewline
\hline 
27 & 568050 & 566244 & 564756 & 553704 & 560820 & 555756\tabularnewline
\hline
\hline
\end{tabular}
\par\end{centering}

\caption{Comparison of  $\TS\left<\rand,-;2\rho,\C\right>$  and $\TS\left<\rand,-;3\rho,\C\right>$ algorithms.}
\label{tab:c-2-6}
\end{table}

For the Frobenius group action, we show results only for $v\in\{3,5\}$ in Table \ref{tab:f-1-6}.
The column  \textbf{bound} shows the upper bounds obtained from  Equation (\ref{eq:frobenius}).

\begin{table}[p]
\begin{centering}
\begin{tabular}{c|c|c|c|c|c|c}
\hline 
$k$ & tab & bound & na\"ive & greedy & col & den\tabularnewline
\hline 
\hline 
\multicolumn{7}{c}{$t=6,v=3$}\\
\hline 
\hline 
53 & 13021 & 13034 & 13029 & 12393 & 12387 & 12393\tabularnewline
\hline 
54 & 14155 & 13120 & 13071 & 12465 & 12513 & 12531\tabularnewline
\hline 
55 & 17161 & 13203 & 13179 & 12561 & 12549 & 12567\tabularnewline
\hline 
56 & 19033 & 13286 & 13245 & 12633 & 12627 & 12639\tabularnewline
\hline 
57 & 20185 & 13366 & 13365 & 12723 & 12717 & 12735\tabularnewline
\hline 
\hline 
\multicolumn{7}{c}{$t=6,v=5$}\\
\hline 
\hline 
31 & 233945 & 226570 & 226425 & 213025 & 212865 & 212865\tabularnewline
\hline 
32 & 258845 & 229820 & 229585 & 216225 & 216085 & 216065\tabularnewline
\hline 
33 & 281345 & 232950 & 232725 & 219285 & 219205 & 219145\tabularnewline
\hline 
34 & 293845 & 235980 & 234905 & 222265 & 223445 & 223265\tabularnewline
\hline 
35 & 306345 & 238920 & 238185 & 225205 & 227445 & 227065\tabularnewline
\hline 
36 & 356045 & 241760 & 241525 & 227925 & 231145 & 230645\tabularnewline
\hline 
\hline 
\end{tabular}
\par\end{centering}

\caption{Comparison of  $\TS\left<\rand,-;\rho,\F\right>$ algorithms.}
\label{tab:f-1-6}
\end{table}

Table \ref{tab:f-2-6} presents results for Frobenius group action  algorithms when the number of maximum uncovered interactions in the first stage is  $2\rho$ or $3\rho$.

\begin{table}[htbp]
\begin{centering}
\begin{tabular}{c|c|c|c||c|c|c}
\hline 
$k$ & \multicolumn{1}{c}{} & \multicolumn{1}{c}{$2\rho$} &  & \multicolumn{1}{c}{} & \multicolumn{1}{c}{$3\rho$} & \tabularnewline
\cline{2-7} 
 & greedy & col & den & greedy & col & den\tabularnewline
\hline 
\hline 
\multicolumn{7}{c}{$t=6,v=3$}\\
\hline 
\hline 
53 & 11931 & 11919 & 11931 & 11700 & 11691 & 11694\tabularnewline
\hline 
54 & 12021 & 12087 & 12087 & 11790 & 11874 & 11868\tabularnewline
\hline 
55 & 12105 & 12237 & 12231 & 11862 & 12057 & 12027\tabularnewline
\hline 
56 & 12171 & 12171 & 12183 & 11949 & 11937 & 11943\tabularnewline
\hline 
57 & 12255 & 12249 & 12255 & 12027 & 12021 & 12024\tabularnewline
\hline 
70 & 13167 & 13155 & 13179 & - & - & - \tabularnewline
\hline 
75 & 13473 & 13473 & 13479 & - & - & - \tabularnewline
\hline 
80 & 13773 & 13767 & 13779 & - & - & - \tabularnewline
\hline 
85 & 14031 & 14025 & 14037 & - & - & - \tabularnewline
\hline 
90 & 14289 & 14283 & 14301 & - & - & - \tabularnewline
\hline 
\hline 
\multicolumn{7}{c}{$t=6,v=5$}\\
\hline 
\hline 
31 & 203785 & 203485 & 203225 & 198945 & 198445 & 197825\tabularnewline
\hline 
32 & 206965 & 208965 & 208065 & 201845 & 204505 & 203105\tabularnewline
\hline 
33 & 209985 & 209645 & 209405 & 205045 & 209845 & 207865\tabularnewline
\hline 
34 & 213005 & 214825 & 214145 & 208065 & 207545 & 206985\tabularnewline
\hline 
35 & 215765 & 215545 & 215265 & 210705 & 210365 & 209885\tabularnewline
\hline 
36 & 218605 & 218285 & 218025 & 213525 & 213105 & 212645\tabularnewline
\hline 
50 & 250625 & 250365 & 250325 & - & - & - \tabularnewline
\hline 
55 & 259785 & 259625 & 259565 & - & - & - \tabularnewline
\hline 
60 & 268185 & 268025 & 267945 & - & - & - \tabularnewline
\hline 
65 & 275785 & 275665 & 275665 & - & - & - \tabularnewline
\hline 
\hline 
\end{tabular}
\par\end{centering}

\caption{Comparison of  $\TS\left<\rand,-;2\rho,\F\right>$  and $\TS\left<\rand,-;3\rho,\F\right>$ algorithms.}
\label{tab:f-2-6}
\end{table}

Next we present a handful of results when $t=5$. 
In the cases examined, using the trivial group action is too time consuming to be practical.
However, the cyclic or Frobenius cases are feasible.
Tables \ref{tab:f-2-5-5} and \ref{tab:c-2-5-6} compare  two stage algorithms when the number of uncovered interactions in the first stage is at most $2\rho$. 

\begin{table}[htbp]
\begin{centering}
\begin{tabular}{c|c|c|c|c}
\hline 
$k$ & tab & greedy & col & den\tabularnewline
\hline 
\hline 
67 & 59110 & 48325 & 48285 & 48305\tabularnewline
\hline 
68 & 60991 & 48565 & 48565 & 48585\tabularnewline
\hline 
69 & 60991 & 48765 & 49005 & 48985\tabularnewline
\hline 
70 & 60991 & 49005 & 48985 & 49025\tabularnewline
\hline 
71 & 60991 & 49245 & 49205 & 49245\tabularnewline
\hline 
\end{tabular}
\par\end{centering}

\caption{Comparison of  $\TS\left<\rand,-;2\rho,\F\right>$ algorithms. $t=5,v=5$}
\label{tab:f-2-5-5}
\end{table}

\begin{table}[htbp]
\begin{centering}
\begin{tabular}{c|c|c|c|c}
\hline 
$k$ & tab & greedy & col & den\tabularnewline
\hline 
\hline 
49 & 122718 & 108210 & 108072 & 107988\tabularnewline
\hline 
50 & 125520 & 109014 & 108894 & 108822\tabularnewline
\hline 
51 & 128637 & 109734 & 110394 & 110166\tabularnewline
\hline 
52 & 135745 & 110556 & 110436 & 110364\tabularnewline
\hline 
53 & 137713 & 111306 & 111180 & 111120\tabularnewline
\hline 
\end{tabular}
\par\end{centering}

\caption{Comparison of  $\TS\left<\rand,-;2\rho,\C\right>$ algorithms. $t=5,v=6$}
\label{tab:c-2-5-6}
\end{table}

In almost all cases there is no clear winner among the three  second stage methods.
Methods \den\  and \greedy\  are, however, substantially faster and use less memory than method \color; for practical purposes they would be preferred.

All  code used in this experimentation is available from the github repository
\begin{center}  https://github.com/ksarkar/CoveringArray \end{center} under an open source GPLv3 license.

\section{Limited dependence and the Moser-Tardos algorithm}\label{sec:mt-algo}
Here we explore a different randomized algorithm that produces smaller covering arrays than Algorithm \ref{algo:rand-slj}. 
When $k>2t$, there are interactions that share no column. 
The events of coverage of such interactions are independent. 
Moser et al. \cite{moser09,moser10} provide an efficient randomized construction method that exploits this limited dependence. 
Specializing their method to covering arrays, we obtain Algorithm \ref{algo:m-t}.
For the specified value of $N$ in the algorithm it is guaranteed that the expected number of times the loop in line \ref{line:mt-loop} of Algorithm \ref{algo:m-t} is repeated is linearly bounded in $k$ (See Theorem 1.2 of \cite{moser10}).

\begin{algorithm}[htbp]
\SetKw{Break}{break}
\KwIn{$t$ : strength of the  covering array, $k$ : number of factors, $v$ : number of levels for each factor}
\KwOut{$A$ : a $\CA(N;t,k,v)$}
Let $N := \frac{\log\left\{{k \choose t} - {k-t \choose t}\right\}+t.\log v+1}{\log\left(\frac{v^{t}}{v^{t}-1}\right)}$\; \label{line:gss}
Construct an $N \times k$ array $A$ where each entry is chosen independently and uniformly at random from a $v$-ary alphabet\; 
\Repeat {covered $=$ true}{ \label{line:mt-loop}
    Set \emph{covered}$:=$ true\;
    \For {each interaction $\iota \in \mathcal{I}_{t,k,v}$}{ \label{line:online-check}
        \If {$\iota$ is not covered in $A$} {
            Set \emph{covered}$:=$ false\;
            Set \emph{missing-interaction} $:= \iota$\;
            \Break\;
        }
    }
    \If {covered $=$ false}{
        Choose all the entries in the $t$ columns involved in \emph{missing-interaction} independently and uniformly at random from the $v$-ary alphabet\;
    }
}
Output $A$\;

\caption{Moser-Tardos type algorithm for covering array construction.}
\label{algo:m-t}
\end{algorithm}

The  upper bound on $\CAN(t,k,v)$  guaranteed by Algorithm \ref{algo:m-t} is obtained by applying the Lov\'asz local lemma (LLL).

\begin{lem}\label{lem:lllsym}
(Lov\'asz local lemma; symmetric case) (see {\rm \cite{alon08}}) 
Let $A_{1},A_{2},\ldots,A_{n}$ events in an arbitrary probability space. 
Suppose that each event $A_{i}$ is mutually independent of a set of all other events $A_{j}$ except for at most $d$, and that $\Pr[A_{i}]\le p$ for all $1\le i\le n$.
If $ep(d+1)\le1$, then $\Pr[\cap_{i=1}^{n}\bar{A_{i}}]>0$.
\end{lem}

The symmetric version of Lov\'asz local lemma provides an upper bound on the probability of a ``bad'' event in terms of the maximum degree of a bad event in a dependence graph, so that the probability that all the bad events are avoided is non zero.
Godbole, Skipper, and Sunley \cite{GSS} apply Lemma \ref{lem:lllsym} essentially to obtain  the  bound on $\CAN(t,k,v)$  in line \ref{line:gss} of Algorithm \ref{algo:m-t}.

\begin{thm}\label{thm:godbole}{\rm \cite{GSS}} Let $t,\, v$ and $k \ge 2t$ be integers with $t,v\ge2$. Then
\[
CAN(t,k,v) \le \frac{\log\left\{{k \choose t} - {k-t \choose t}\right\}+t \log v+1}{\log\left(\frac{v^{t}}{v^{t}-1}\right)}
\]
\end{thm}

The bound on the size of covering arrays obtained from Theorem \ref{thm:godbole} is asymptotically tighter than the one obtained from Theorem \ref{thm:slj}.
Figure \ref{fig:comp-slj-gss} compares the bounds for $t=6$ and $v=3$.  

The original proof of LLL is essentially non-constructive and does not immediately lead to a polynomial time construction algorithm for covering arrays satisfying the bound of Theorem \ref{thm:godbole}. 
Indeed no previous construction algorithms appear to be based on it.
However the Moser-Tardos method of Algorithm \ref{algo:m-t} does provide a construction algorithm running in expected polynomial time.
For sufficiently large values of $k$ Algorithm \ref{algo:m-t} produces smaller covering arrays than the Algorithm \ref{algo:rand-slj}.

\begin{figure}[htb]
\begin{centering}
\includegraphics[bb=100bp 238bp 520bp 553bp,clip,scale=0.8]{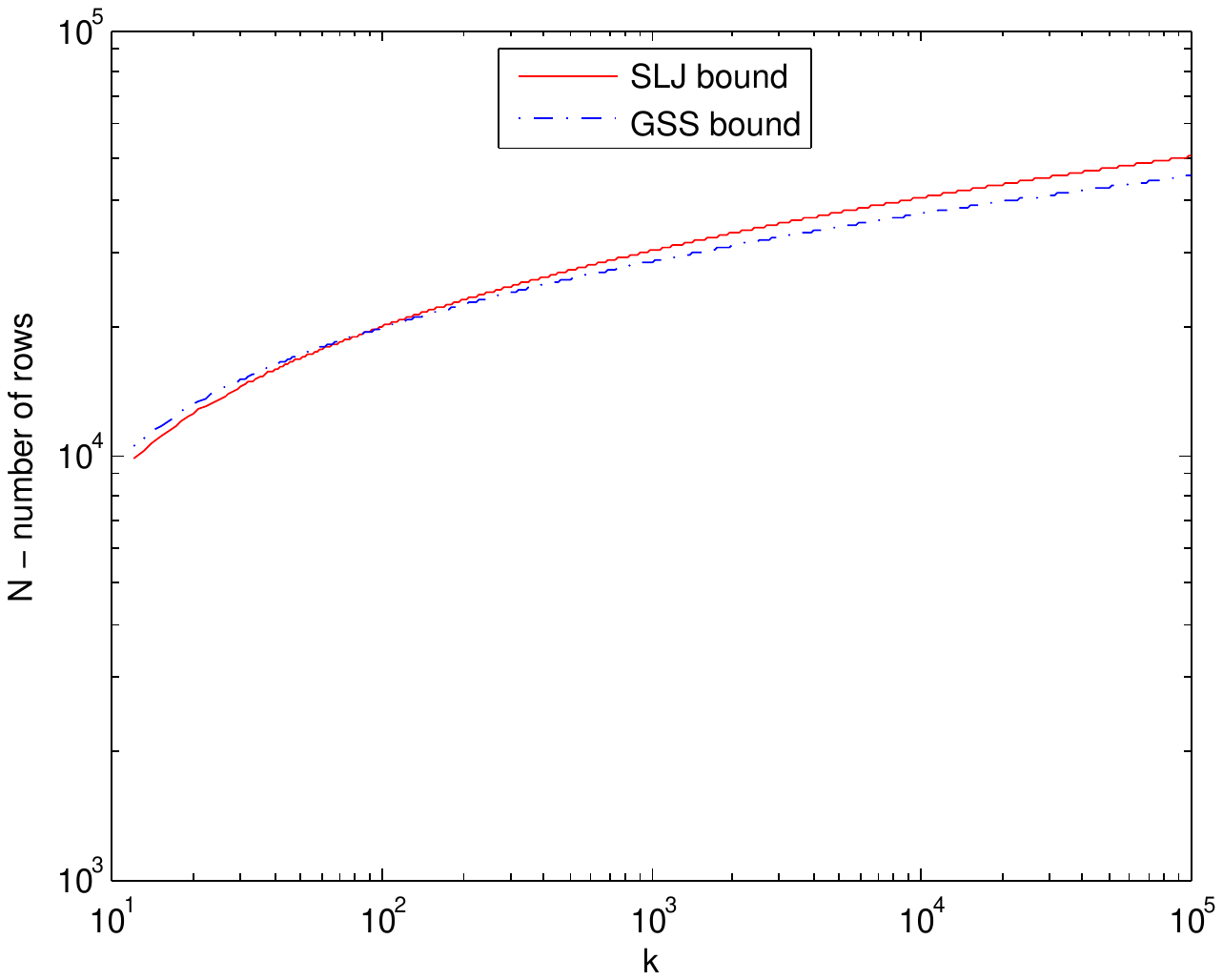}
\par\end{centering}

\caption{Comparison of SLJ (Theorem \ref{thm:slj}) and GSS (Theorem \ref{thm:godbole}) bounds  for $t=6$ and $v=3$. The graph is plotted in log-log scale to highlight the asymptotic difference between the two bounds.}
\label{fig:comp-slj-gss}
\end{figure}

But the question remains:  Does Algorithm \ref{algo:m-t}  produce smaller covering arrays than the currently best known results within the range that it can be effectively computed?
Perhaps surprisingly, we show that the answer is affirmative.
In Algorithm \ref{algo:m-t} we do not need to store the coverage information of  individual interactions in memory because each time an uncovered interaction is encountered  we re-sample the columns involved in that interaction and start the check afresh (checking the coverage in interactions in the same order each time). 
Consequently, Algorithm \ref{algo:m-t} can be applied for larger values of $k$ than the density algorithm.

Smaller covering arrays can be obtained by exploiting a group action using LLL, as shown in \cite{sarkar16}.
Table \ref{tab:lll} shows the sizes of the covering arrays constructed by a variant of Algorithm \ref{algo:m-t} that employs cyclic and Frobenius group actions.
While this single stage algorithm produces smaller arrays than the currently best known  \cite{catables}, these are already superseded by the two-stage based algorithms.


\begin{table}[htbp]
\begin{centering}
\begin{subtable}[h]{0.3\textwidth}
\begin{centering}
\begin{tabular}{c|c|c}
\hline 
k & tab & MT\tabularnewline
\hline 
\hline 
56 & 19033 & 16281\tabularnewline
\hline 
57 & 20185 & 16353\tabularnewline
\hline 
58 & 23299 & 16425\tabularnewline
\hline 
59 & 23563 & 16491\tabularnewline
\hline 
60 & 23563 & 16557\tabularnewline
\hline 
\end{tabular}
\caption{\F. $t=6,\, v=3$}
\par\end{centering}
\end{subtable}
\quad
\begin{subtable}[h]{0.3\textwidth}
\begin{centering}
\begin{tabular}{c|c|c}
\hline 
k & tab & MT\tabularnewline
\hline 
\hline 
44 & 411373 & 358125\tabularnewline
\hline 
45 & 417581 & 360125\tabularnewline
\hline 
46 & 417581 & 362065\tabularnewline
\hline 
47 & 423523 & 363965\tabularnewline
\hline 
48 & 423523 & 365805\tabularnewline
\hline 
\end{tabular}
\caption{\F. $t=6,\, v=5$}
\par\end{centering}
\end{subtable}
\quad
\begin{subtable}[h]{0.3\textwidth}
\begin{centering}
\begin{tabular}{c|c|c}
\hline 
k & tab & MT\tabularnewline
\hline 
\hline 
25 & 1006326 & 1020630\tabularnewline
\hline 
26 & 1040063 & 1032030\tabularnewline
\hline 
27 & 1082766 & 1042902\tabularnewline
\hline 
28 & 1105985 & 1053306\tabularnewline
\hline 
29 & 1149037 & 1063272\tabularnewline
\hline 
\end{tabular}
\caption{\C. $t=6,\, v=6$}
\par\end{centering}
\end{subtable}
\par\end{centering}

\caption{Comparison of covering array size  from Algorithm \ref{algo:m-t} (MT) with
the best known results  \cite{catables} (tab).}
\label{tab:lll}
\end{table}

\subsection{Moser-Tardos type algorithm for the first stage}\label{subsec:lll-first-stage}
The linearity of expectation arguments used in the SLJ bounds permit one to consider situations in which a few of the ``bad'' events are allowed to occur, a fact that we exploited in the first stage of the algorithms thus far.
However, the Lov\'asz local lemma does not address this situation directly.  
The conditional Lov\'asz local lemma (LLL) distribution, introduced in \cite{haeupler10}, is a very useful tool.

\begin{lem}\label{lem:condlll}
(Conditional LLL distribution; symmetric case) (see \cite{alon08,sarkar16}) 
Let $\mathscr{A} = \{A_{1},A_{2},\ldots,A_{l}\}$ be a set of $l$ events in an arbitrary probability space. 
Suppose that each event $A_{i}$ is mutually independent of a set of all other events $A_{j}$ except for at most $d$, and that $\Pr[A_{i}]\le p$ for all $1\le i\le l$.
Also suppose that $ep(d+1)\le 1$ (Therefore, by LLL (Lemma \ref{lem:lllsym}) $\Pr[\cap_{i=1}^{l}\bar{A_{i}}]>0$).
Let $B\notin \mathscr{A}$ be another event in the same probability space with $\Pr[B] \le q$, such that $B$ is also mutually independent of a set of all other events $A_{j} \in \mathscr{A}$ except for at most $d$.
Then $\Pr[B|\cap_{i=1}^{l}\bar{A_{i}}] \le eq$.
\end{lem}

We  apply the conditional LLL distribution to obtain an upper bound on the size of partial array that leaves at most $\log\left(\frac{v^{t}}{v^{t}-1}\right) \approx v^t$ interactions uncovered.
For a positive integer $k$,  let $I = \{j_1, \ldots, j_\rho\} \subseteq [k]$ where $j_1<\ldots<j_\rho$. 
Let $A$ be an $n \times k$ array where each entry is from the set $[v]$. 
Let $A_I$ denote the $n \times \rho$ array in which $A_I(i,\ell) = A(i,j_\ell)$ for $1 \le i \le N$ and $1 \le \ell \le \rho$; $A_I$ is the projection of $A$ onto the columns in $I$.

Let $M\subseteq [v]^t$ be a set of $m$ $t$-tuples of symbols, and $C\in \binom{[k]}{t}$ be a set of $t$ columns.
Suppose the entries in the array $A$ are chosen independently from $[v]$ with uniform probability.
Let $B_C$ denote the event that at least one of the tuples in $M$ is not covered in $A_C$.
There are $\eta = \binom{k}{t}$ such events, and for all of them $\Pr[B_C]\le m\left(1-\frac{1}{v^t}\right)^n$.
Moreover, when $k\ge 2t$, each of the events is mutually independent of all other events except for at most $\rho = \binom{k}{t} - \binom{k-t}{t} - 1 < t\binom{k}{t-1}$.
Therefore, by the Lov\'asz local lemma, when $e\rho m \left(1-\frac{1}{v^t}\right)^n \le 1$, none of the events $B_C$ occur.
Solving for $n$, when
\begin{equation}\label{eq:n1}
n \ge \frac{\log (e\rho m)}{\log\left(\frac{v^{t}}{v^{t}-1}\right)}
\end{equation}
there exists an $n \times k$ array $A$ over $[v]$ such that for all $C\in \binom{[k]}{t}$, $A_C$ covers all the $m$ tuples in $M$.
In fact we can use a Moser-Tardos type algorithm to construct such an array.

Let $\iota$ be an interaction whose $t$-tuple of symbols is not in $M$.
Then the probability that $\iota$ is not covered in an $n \times k$ array is at most $\left(1-\frac{1}{v^t}\right)^n$ when each entry of the array is chosen independently from $[v]$ with uniform probability.
Therefore, by the conditional LLL distribution the probability that $\iota$ is not covered in the array $A$ where for all $C\in \binom{[k]}{t}$, $A_C$ covers all the $m$ tuples in $M$ is at most $e\left(1-\frac{1}{v^t}\right)^n$.
Moreover, there are $\eta (v^t - m)$ such interactions $\iota$.
By the linearity of expectation, the expected number of uncovered interactions in $A$ is less than $v^t$ when $\eta (v^t - m) e\left(1-\frac{1}{v^t}\right)^n \le v^t$.
Solving for $n$, we obtain
\begin{equation}\label{eq:n2}
n \ge \frac{\log \left\{\eta e \left(1 - \frac{m}{v^t}\right)\right\}}{\log\left(\frac{v^{t}}{v^{t}-1}\right)}.
\end{equation}

Therefore, there exists an $n \times k$ array with $n = \max\left\{\frac{\log (e\rho m)}{\log\left(\frac{v^{t}}{v^{t}-1}\right)}, \frac{\log \left\{\eta e \left(1 - \frac{m}{v^t}\right)\right\}}{\log\left(\frac{v^{t}}{v^{t}-1}\right)} \right\}$ that has at most $v^t$ uncovered interactions.
To compute $n$ explicitly, we must choose $m$. 
We can select a value of $m$ to minimize $n$  graphically for given values of $t,\,k$ and $v$.
For example, Figure \ref{fig:mt-fs} plots Equations \ref{eq:n1} and \ref{eq:n2} against $m$ for $t=3,\,k=350,\,v=3$, and finds the minimum value of $n$.

\begin{figure}[htb]
    \centering
    \begin{subfigure}[b]{0.4\textwidth}
        \includegraphics[bb=100bp 238bp 520bp 553bp,clip,scale=0.5]{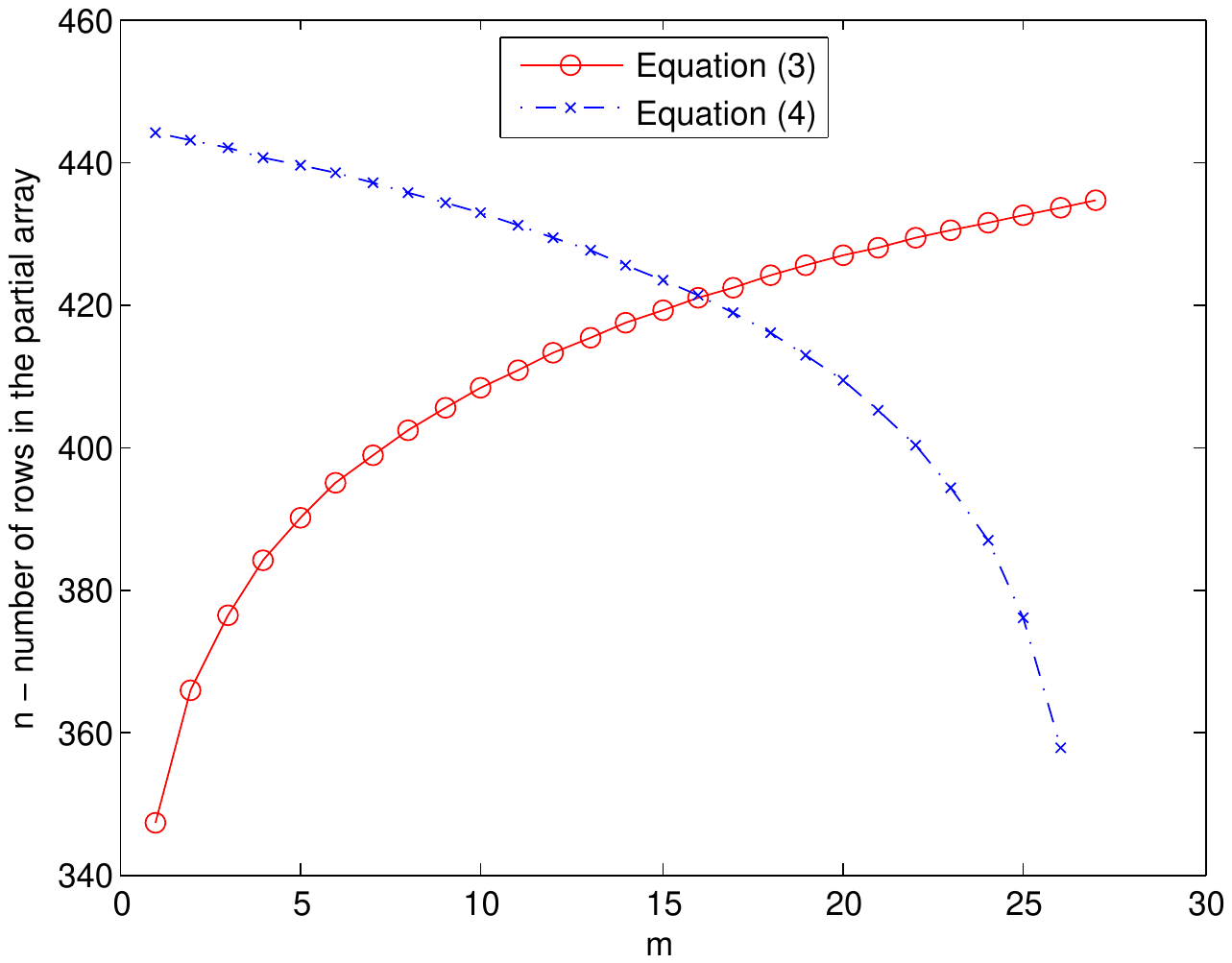}
        \caption{Equations \ref{eq:n1} and \ref{eq:n2} against $m$.}
    \end{subfigure}
    \quad 
    \begin{subfigure}[b]{0.4\textwidth}
        \includegraphics[bb=100bp 238bp 520bp 553bp,clip,scale=0.5]{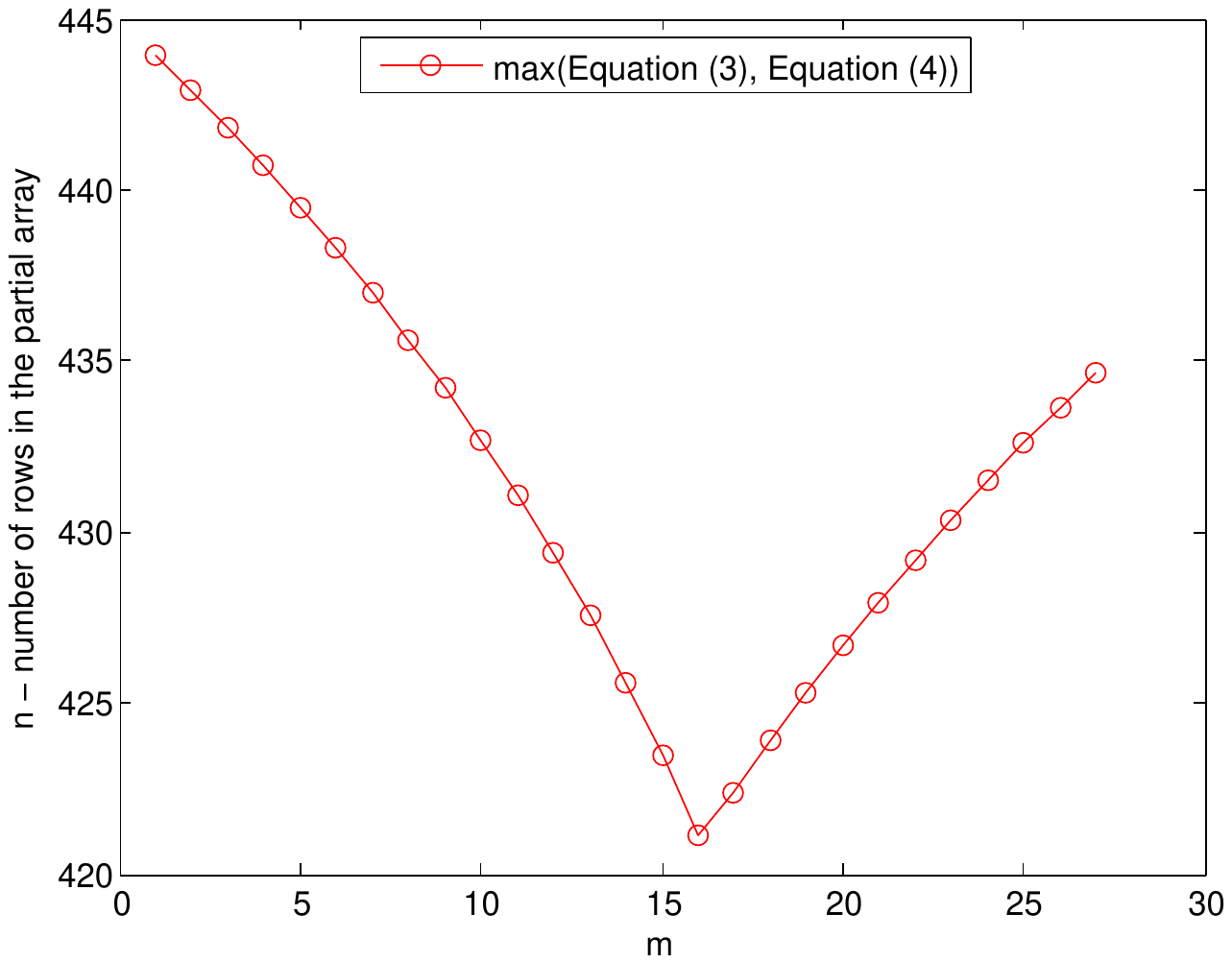}
        \caption{Maximum of the two sizes against $m$.}
    \end{subfigure}
    \caption{The minimum  is at $n=422$,  when $m=16$.
    $t=3,\,k=350,\,v=3$}\label{fig:eq-comp}
 	\label{fig:mt-fs}
\end{figure}

We compare the size of the partial array from the na\"ive two-stage method (Algorithm \ref{algo:two-stage-simple}) with the size obtained by the graphical methods in Figure \ref{fig:mt-rand-comp}.
The Lov\'asz local lemma based method is asymptotically better than the simple randomized method. 
However, except for the small values of $t$ and $v$, in the range of $k$ values  relevant for practical applications the simple randomized algorithm requires fewer rows than the Lov\'asz local lemma based method.

\begin{figure}[htb]
    \centering
    \begin{subfigure}[b]{0.4\textwidth}
        \includegraphics[bb=100bp 238bp 520bp 553bp,clip,scale=0.5]{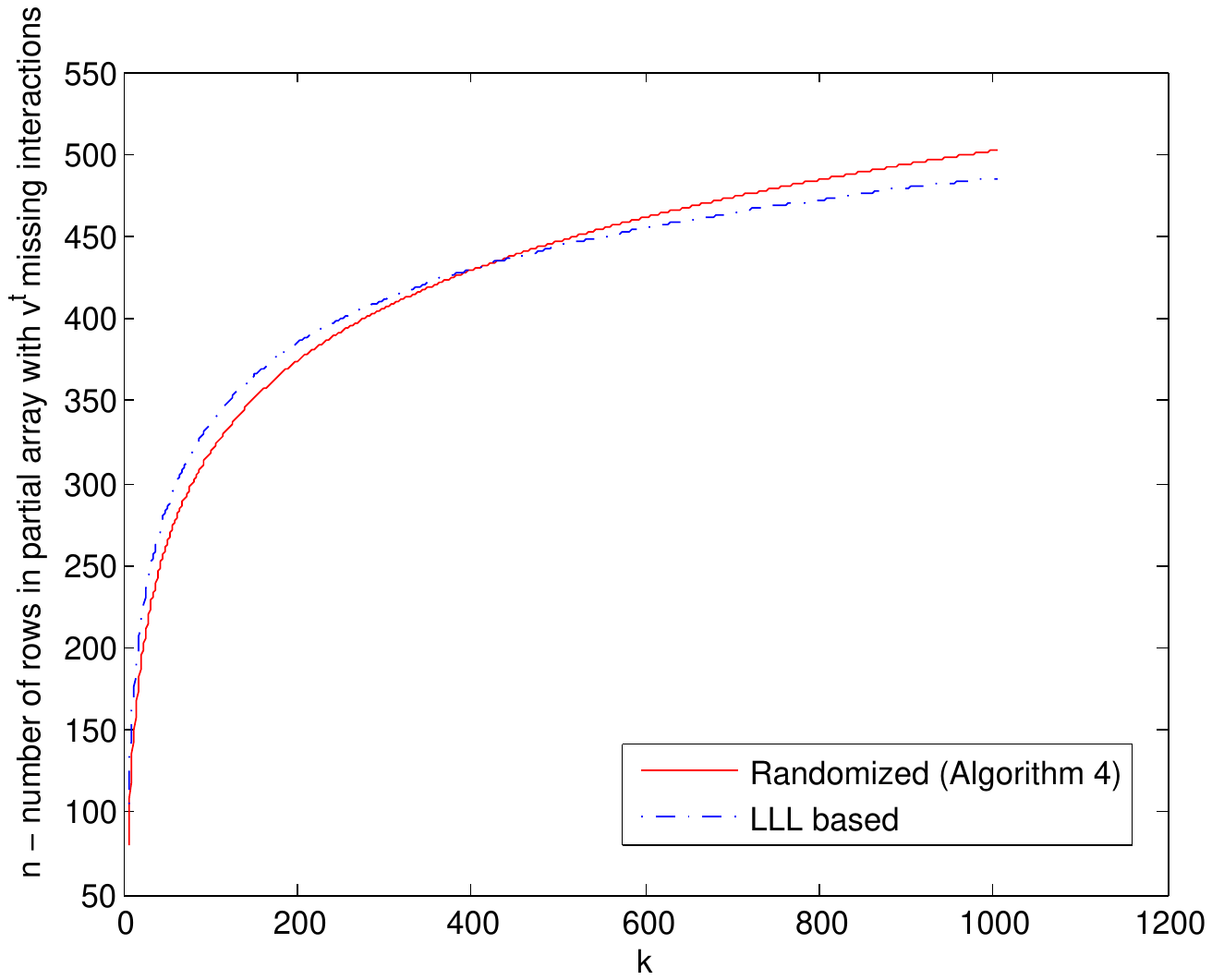}
        \caption{$t=3,\,v=3$.}
    \end{subfigure}
    \quad 
    \begin{subfigure}[b]{0.4\textwidth}
        \includegraphics[bb=100bp 238bp 520bp 553bp,clip,scale=0.5]{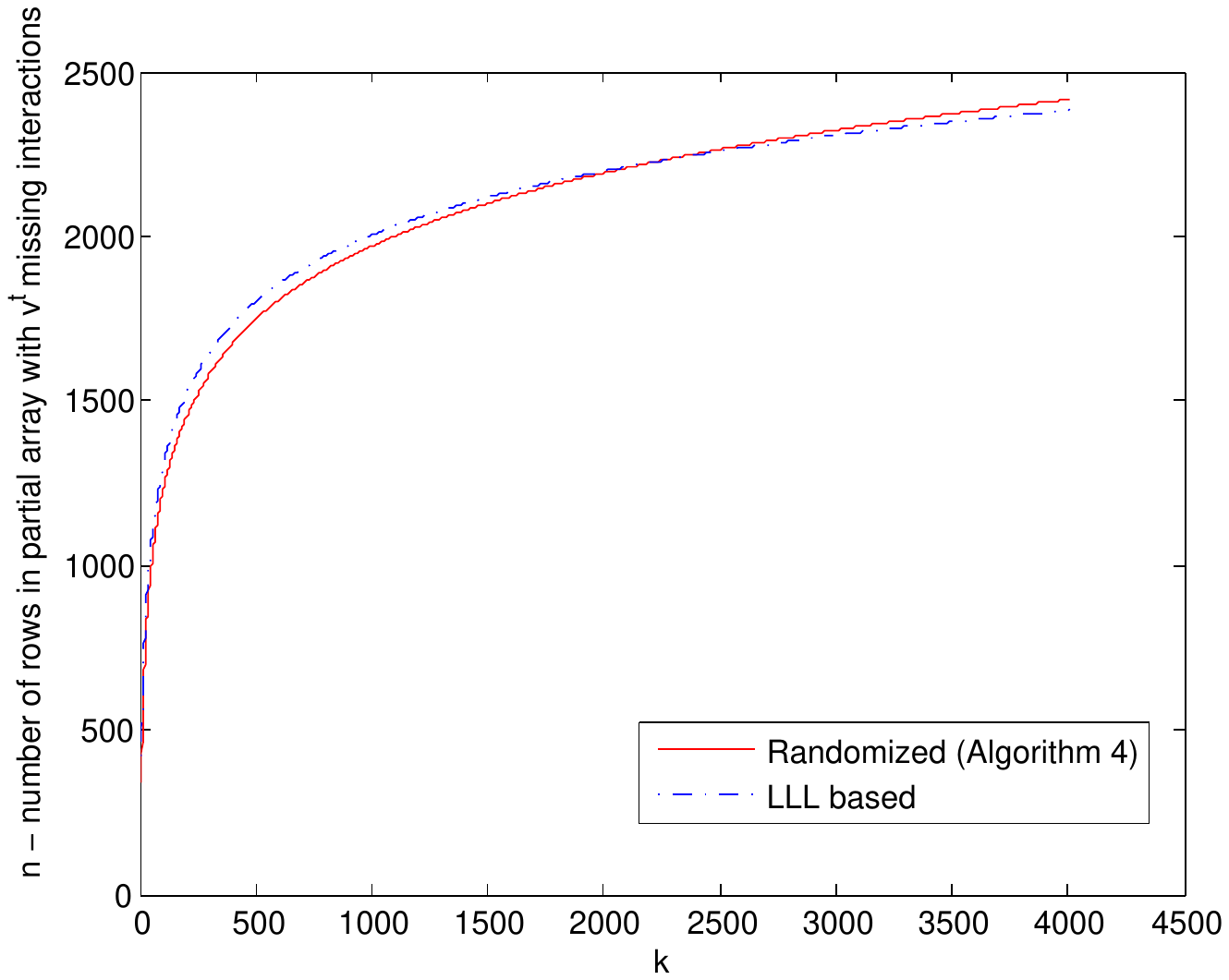}
        \caption{$t=4,\,v=3$.}
    \end{subfigure}
    \caption{Comparison of the size of the partial array constructed in the first stage.
    Compares the size of the partial array specified in Algorithm \ref{algo:two-stage-simple} in Section \ref{subsec:simple}, and the size derived in Section \ref{subsec:lll-first-stage}. }\label{fig:pa-size}
    \label{fig:mt-rand-comp}
\end{figure}

\subsection{Lov\'asz local lemma based two-stage bound}\label{subsec:lll-two-stage}
We can apply the techniques from Section \ref{subsec:lll-first-stage} to obtain a two-stage bound similar to Theorem \ref{thm:two-stage} using the Lov\'asz local lemma and conditional LLL distribution.
First we  extend a result from \cite{sarkar16}.

\begin{thm} \label{thm:lll-two-stage}
Let $t,\, k,\, v$ be integers with $k\ge t\ge2$, $v\ge2$ and let $\eta = \binom{k}{t}$, and $\rho = \binom{k}{t} - \binom{k-t}{t}$. 
If $\frac{\eta v^t \log\left(\frac{v^{t}}{v^{t}-1}\right)}{\rho} \le v^t$
Then 
\[ 
\CAN(t,k,v)\le \frac{\log{k \choose t}+t\log v+\log\log\left(\frac{v^{t}}{v^{t}-1}\right)+2}{\log\left(\frac{v^{t}}{v^{t}-1}\right)} - \frac{\eta}{\rho}.
\] 
\end{thm}

\begin{proof}
Let $M\subseteq [v]^t$ be a set of $m$ $t$-tuples of symbols.
Following the arguments of Section \ref{subsec:lll-first-stage},  when $n \ge \frac{\log (e\rho m)}{\log\left(\frac{v^{t}}{v^{t}-1}\right)}$  there exists an $n \times k$ array $A$ over $[v]$ such that for all $C\in \binom{[k]}{t}$, $A_C$ covers all  $m$ tuples in $M$.

At most $\eta (v^t-m)$ interactions are uncovered in such an array.
Using the conditional LLL distribution, the probability that one such interaction is not covered in $A$ is at most $e\left(1-\frac{1}{v^t}\right)^n$.
Therefore, by the linearity of expectation, we can find one such array $A$ that leaves at most $e \eta (v^t-m) \left(1-\frac{1}{v^t}\right)^n = \frac{\eta}{\rho}\left(\frac{v^t}{m}-1\right)$ interactions uncovered.
Adding one row per uncovered interactions to $A$, we obtain a covering array with at most $N$ rows, where
\[
N = \frac{\log (e\rho m)}{\log\left(\frac{v^{t}}{v^{t}-1}\right)} + \frac{\eta}{\rho}\left(\frac{v^t}{m}-1\right)
\]

The value of $N$ is minimized when $m= \frac{\eta v^t \log\left(\frac{v^{t}}{v^{t}-1}\right)}{\rho}$.
Because  $m \le v^t$, we obtain the desired bound.
\end{proof}

When $m=v^t$ this recaptures  the  bound of Theorem \ref{thm:godbole}.

Figure \ref{fig:bounds-comp} compares the LLL based two-stage bound  from Theorem \ref{thm:lll-two-stage} to the standard two-stage bound  from Theorem \ref{thm:two-stage}, the Godbole et al. bound in Theorem \ref{thm:godbole}, and   the SLJ bound  in Theorem \ref{thm:slj}.
Although the LLL based two-stage bound is tighter than the LLL based Godbole et al. bound, even for quite large values of $k$ the standard two-stage bound is tighter than the LLL based two-stage bound.
In practical terms, this specific LLL based two-stage method does not look very promising, unless the parameters are quite large.

\begin{figure}[htb]
    \centering
    \begin{subfigure}[b]{0.4\textwidth}
        \includegraphics[bb=100bp 238bp 520bp 553bp,clip,scale=0.5]{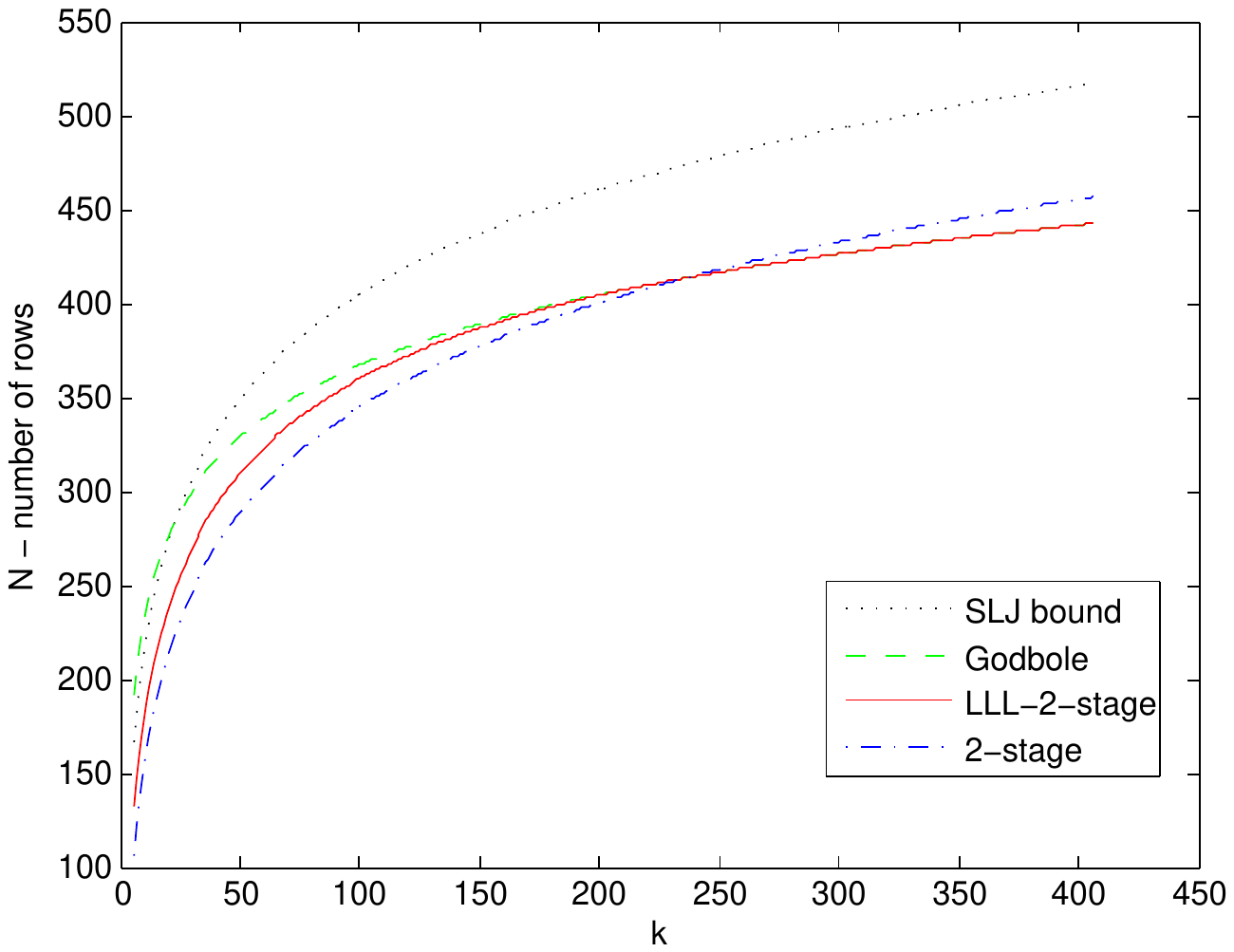}
        \caption{$t=3,\,v=3$.}
    \end{subfigure}
    \quad 
    \begin{subfigure}[b]{0.4\textwidth}
        \includegraphics[bb=100bp 238bp 520bp 553bp,clip,scale=0.5]{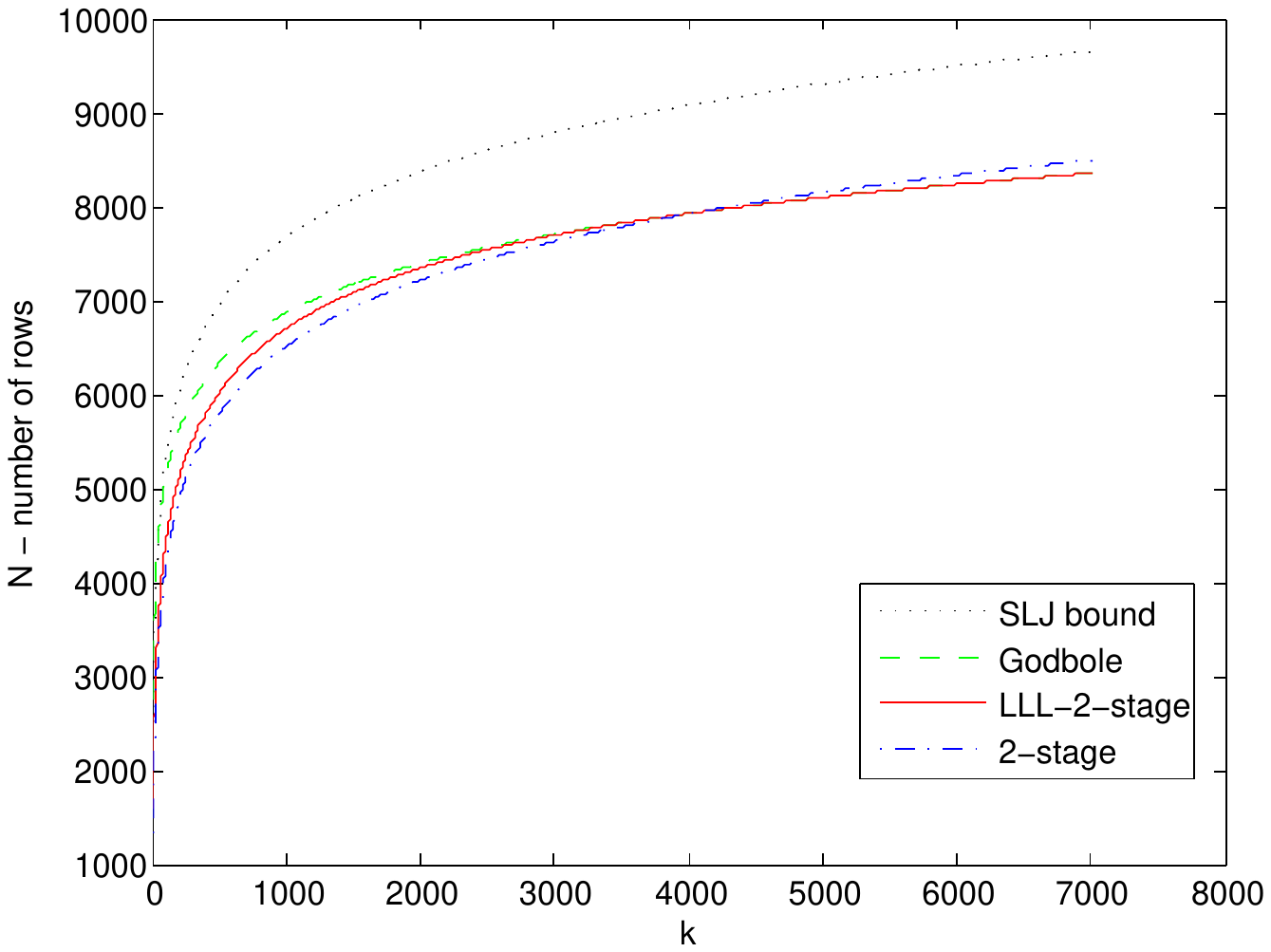}
        \caption{$t=4,\,v=4$.}
    \end{subfigure}
    \caption{Comparison among the LLL based two-stage bound  from Theorem \ref{thm:lll-two-stage}, the standard two-stage bound  from Theorem \ref{thm:two-stage}, the Godbole et al. bound in Theorem \ref{thm:godbole}, and   the SLJ bound  in Theorem \ref{thm:slj}.}\label{fig:bounds-comp}
\end{figure}

\section{Conclusion and open problems}\label{sec:conc}
Many concrete algorithms within a two-stage framework for covering array construction have been introduced and evaluated. 
The two-stage approach extends the range of parameters for which competitive covering arrays can be constructed, each meeting an \emph{a priori} guarantee on its size.
Indeed as a consequence a number of best known covering arrays have been improved upon.
Although each of the methods proposed has useful features, our experimental evaluation suggests that $\TS\left<\rand, \greedy; 2\rho,\Gamma\right>$ and $\TS\left<\rand, \den; 2\rho.\Gamma\right>$ with $\Gamma \in \{\C, \F\}$ realize a good trade-off between running time and size of the constructed covering array.

Improvements in the bounds, or in the algorithms that realize them,  are certainly of interest.
We mention some specific directions.
Establishing tighter bounds on the coloring based methods of Section \ref{subsec:coloring} is a challenging problem. 
Either better estimates of the chromatic number of the incompatibility graph after a random first stage, or a first stage designed to limit the chromatic number, could lead to improvements in the bounds.

In Section \ref{subsec:lll-first-stage} and \ref{subsec:lll-two-stage} we  explored using a Moser-Tardos type algorithm for the first stage.
Although this is useful for asymptotic bounds, practical improvements appear to be limited.
Perhaps a different approach of reducing the number of bad events to be avoided explicitly by the algorithm may lead to a better algorithm.
A potential approach may look like following: ``Bad'' events would denote non-coverage of an interaction over a $t$-set of columns. 
We would select a set of column $t$-sets such that the dependency graph of the corresponding bad events have a bounded maximum degree (less than the original dependency graph).
We would devise a Moser-Tardos type algorithm for covering all the interactions on our chosen column $t$-sets, and  then apply the conditional LLL distribution to obtain an upper bound on the number of uncovered interactions.
However, the difficulty lies in the fact that ``all vertices have degree $\le\rho$'' is a non-trivial, ``hereditary'' property for induced subgraphs, and for such properties finding a maximum induced subgraph with the property is an NP-hard optimization problem \cite{Garey}.
There is still hope for a randomized or ``nibble'' like strategy that may find a reasonably good induced subgraph with a bounded maximum degree.
Further exploration of this idea seems to be a promising research avenue.

In general, one could consider more than  two stages. 
Establishing the benefit (or not) of having more than two stages is also an interesting open problem.
Finally, the application of the methods developed to mixed covering arrays appears to provide useful techniques for higher strengths; this merits further study as well. 

\section*{Acknowledgments}

The  research was supported in part by the National Science Foundation under Grant No. 1421058.

\bibliographystyle{abbrv}
\bibliography{CoveringArray}
\end{document}